\documentclass[11pt,onecolumn]{article}
\usepackage{amsmath,amsthm,amssymb,amsfonts,epsfig,graphicx}
\usepackage{hyperref}
\usepackage{cleveref}
\usepackage{color}
\usepackage{enumitem}
\usepackage{url}
\usepackage{array,multirow}
\usepackage{bm}

\usepackage{stmaryrd}

\usepackage{enumitem}
\usepackage[vlined,linesnumbered,ruled,resetcount]{algorithm2e}

\usepackage[left=1.0in,top=0.5in,right=1.05in,bottom=1.0in,nohead,textheight=10in,footskip=0.3in]{geometry}

\newtheorem{theorem}{Theorem}
\newtheorem{corollary}[theorem]{Corollary}

\newtheorem{assumption}{Assumption}
\newtheorem{lemma}[theorem]{Lemma}
\newtheorem{definition}[theorem]{Definition}
\newtheorem{proposition}[theorem]{Proposition}
\newtheorem{remark}{Remark}

\SetKwInput{kwMyRepeat}{Repeat}

\begin{document}

\def\myparagraph#1{\vspace{2pt}\noindent{\bf #1~~}}

%\pagestyle{headings}

%%%%%%%%%%%%%%%%%%%%%%%%%%%%%%%%%%%%%%%%%%%%%%%%%%%%%%%%%%%
%%%%%%%%%%%%%%%%%%%%%%%%%%%%%%%%%%%%%%%%%%%%%%%%%%%%%%%%%%%
%%%%%%%%%%%%%%%%%%%%%%%%%%%%%%%%%%%%%%%%%%%%%%%%%%%%%%%%%%%

\def\DeltaCeil{{\lceil\Delta\rceil}}
\def\TwoDeltaCeil{{\lceil 2\Delta\rceil}}
\def\OnePointFiveDeltaCeil{{\lceil 3\Delta/2\rceil}}

\long\def\ignore#1{}
%\epsfverbosetrue
%\def\myps[#1]#2{}
\def\myps[#1]#2{\includegraphics[#1]{#2}}
\def\etal{{\em et al.}}
\def\Bar#1{{\bar #1}}
\def\br(#1,#2){{\langle #1,#2 \rangle}}
\def\setZ[#1,#2]{{[ #1 .. #2 ]}}
\def\Pr#1{\mbox{\tt Pr}\left[{#1}\right]}
\def\REACHED{\mbox{\tt REACHED}}
\def\AdjustFlow{\mbox{\tt AdjustFlow}}
\def\GetNeighbors{\mbox{\tt GetNeighbors}}
\def\true{\mbox{\tt true}}
\def\false{\mbox{\tt false}}
\def\Process{\mbox{\tt Process}}
\def\ProcessLeft{\mbox{\tt ProcessLeft}}
\def\ProcessRight{\mbox{\tt ProcessRight}}
\def\Add{\mbox{\tt Add}}

\newcommand{\eqdef}{{\stackrel{\mbox{\tiny \tt ~def~}}{=}}}

\def\setof#1{{\left\{#1\right\}}}
\def\suchthat#1#2{\setof{\,#1\mid#2\,}} % so says Knuth, page 174
\def\event#1{\setof{#1}}
\def\q={\quad=\quad}
\def\qq={\qquad=\qquad}
\def\calA{{\cal A}}
\def\calB{{\cal B}}
\def\calC{{\cal C}}
\def\calD{{\cal D}}
\def\calE{{\cal E}}
\def\calK{{\cal K}}
\def\calG{{\cal G}}
\def\calI{{\cal I}}
\def\calF{{\cal F}}
\def\calH{{\cal H}}
\def\calL{{\cal L}}
\def\calN{{\cal N}}
\def\calM{{\cal M}}
\def\calP{{\cal P}}
\def\calR{{\cal R}}
\def\calS{{\cal S}}
\def\calT{{\cal T}}
\def\calU{{\cal U}}
\def\calV{{\cal V}}
\def\calO{{\cal O}}
\def\calX{{\cal X}}

\def\U{{\mathbb U}}
\def\E{{\calE}}
\def\T{{\calF}}
\def\XX(#1){{{#1}^\downarrow}}

\def\calY{{\cal Y}}
\def\calZ{{\cal Z}}
\def\s{\footnotesize}
\def\calNG{{\cal N_G}}
\def\psfile[#1]#2{}
\def\psfilehere[#1]#2{}
\def\epsfw#1#2{\includegraphics[width=#1\hsize]{#2}}
\def\assign(#1,#2){\langle#1,#2\rangle}
\def\edge(#1,#2){(#1,#2)}
\def\VS{\calV^s}
\def\VT{\calV^t}
\def\slack(#1){\texttt{slack}({#1})}
\def\barslack(#1){\overline{\texttt{slack}}({#1})}
\def\NULL{\texttt{NULL}}
\def\PARENT{\texttt{PARENT}}
\def\GRANDPARENT{\texttt{GRANDPARENT}}
\def\TAIL{\texttt{TAIL}}
\def\HEADORIG{\texttt{HEAD$\_\:$ORIG}}
\def\TAILORIG{\texttt{TAIL$\_\:$ORIG}}
\def\HEAD{\texttt{HEAD}}
\def\CURRENTEDGE{\texttt{CURRENT$\!\_\:$EDGE}}

\def\unitvec(#1){{{\bf u}_{#1}}}
\def\uvec{{\bf u}}
\def\vvec{{\bf v}}
\def\Nvec{{\bf N}}
\def\r{{\bf r}}

\newcommand{\bg}{\mbox{$\bf g$}}
\newcommand{\bh}{\mbox{$\bf h$}}

\newcommand{\bx}{\mbox{$x$}}
\newcommand{\by}{\mbox{\boldmath $y$}}
\newcommand{\bz}{\mbox{\boldmath $z$}}
\newcommand{\bu}{\mbox{\boldmath $u$}}
\newcommand{\bv}{\mbox{\boldmath $v$}}
\newcommand{\bw}{\mbox{\boldmath $w$}}
\newcommand{\bvarphi}{\mbox{\boldmath $\varphi$}}
\newcommand{\balpha}{\mbox{\boldmath $\alpha$}}

\newcommand\myqed{{}}

\newcommand{\IBFSFS}{{IBFS$^{\mbox{~\!\tiny FS}}$}}
\newcommand{\IBFSAL}{{IBFS$^{\mbox{~\!\tiny AL}}$}}
\newcommand{\BKFS}{{BK$^{\mbox{~\!\tiny FS}}$}}
\newcommand{\BKAL}{{BK$^{\mbox{~\!\tiny AL}}$}}

\newcommand{\XL}{{X_{\le L}}}

\def\br#1{{\llbracket #1 \rrbracket}}

%%%%%%%%%%%%%%%%%%%%%%%%%%%%%%%%%%%%%%%%%%%%%%%%%%%%%%%%%%%
%%%%%%%%%%%%%%%%%%%%%%%%%%%%%%%%%%%%%%%%%%%%%%%%%%%%%%%%%%%
%%%%%%%%%%%%%%%%%%%%%%%%%%%%%%%%%%%%%%%%%%%%%%%%%%%%%%%%%%%

\title{\Large\bf  }
\author{}
\title{
\Large\bf  \vspace{-20pt} A simpler and parallelizable {$O(\sqrt{\log n})$}-approximation \\ algorithm  for {\sc Sparsest Cut}
}
\author{Vladimir Kolmogorov \\ \normalsize Institute of Science and Technology Austria (ISTA) \\ {\normalsize\tt vnk@ist.ac.at}}
\date{}
\maketitle
\begin{abstract}
Currently, the best known tradeoff between approximation ratio and complexity
for the {\sc Sparsest Cut} problem is achieved by the algorithm in [Sherman, FOCS 2009]:
it computes $O(\sqrt{(\log n)/\varepsilon})$-approximation using $O(n^\varepsilon\log^{O(1)}n)$ maxflows
for any $\varepsilon\in[\Theta(1/\log n),\Theta(1)]$.
It works by solving the SDP relaxation of [Arora-Rao-Vazirani, STOC 2004]
using the Multiplicative Weights Update algorithm (MW) of [Arora-Kale, JACM 2016].
To implement one MW step, Sherman approximately solves a multicommodity flow problem
using another application of MW. Nested MW steps are solved via
a certain ``chaining'' algorithm that combines results of multiple calls to the maxflow algorithm.

We present an alternative approach that avoids solving the multicommodity flow problem
and instead computes ``violating paths''.
This simplifies Sherman's algorithm by removing a need for a nested application of MW,
and also allows parallelization: we show how to compute
$O(\sqrt{(\log n)/\varepsilon})$-approximation via $O(\log^{O(1)}n)$ maxflows
using $O(n^\varepsilon)$ processors.

We also revisit Sherman's chaining algorithm, and present a simpler version together with a new analysis.
\end{abstract}

\section{Introduction}
Partitioning a given undirected graph $G=(V,E)$ into two (or more) components is a fundamental problem in computer science with many real-world applications,
ranging from data clustering and network analysis to parallel computing and VLSI design.
Usually, desired partitions should satisfy two properties:
(i) the total cost of edges between different components should be small, and
(ii) the components should be sufficiently balanced.
For partitions with two components $(S,\bar S)$ this means
that $E(S,\bar S)$ should be small
and $\min\{|S|,|\bar S|\}$ should be large,
where $E(S,\bar S)$ is the total number of edges between $S$ and $\bar S$
(or their total weight in the case of weighted graphs).

One of the most widely studied versions is the {\sc Sparsest Cut} problem
whose goal is to minimize the ratio $\frac{E(S,\bar S)}{\min\{|S|,|\bar S|\}}$
(called {\em edge expansion}) over partitions $(S,\bar S)$.
Another well-known variant is the {\sc $c$-Balanced Separator} problem:
 minimize $E(S,\bar S)$
over {\em $c$-balanced partitions}, i.e.\ partitions satisfying ${\min\{|S|,|\bar S|\}}\ge cn$
where $n=|V|$ and $c\in(0,\tfrac 12)$ is a given constant.

Both problems are NP-hard, which forces one to study approximation algorithms,
or pseudoapproximation algorithms in the case of {\sc Balanced Separator}.
(An algorithm for {\sc Balanced Separator} is said to be a $\lambda$-pseudoapproximation
if it computes a $c'$-balanced partition $(S,\bar S)$ whose cost $E(S,\bar S)$
is at most $\lambda$ times the optimal cost of the {\sc $c$-Balanced Partition} problem,
for some  constants $c'\le c$). Below we discuss known results for the {\sc Sparsest Cut} problem.
They also apply to {\sc Balanced Separator}: in all previous works, whenever there is an $O(\lambda)$-approximation algorithm
for {\sc Sparsest Cut} then there is an $O(\lambda)$-pseudoapproximation algorithm for {\sc Balanced Separator}
with the same complexity.

The first nontrivial guarantee was obtained by Leighton and Rao~\cite{LeightonRao},
who presented a $O(\log n)$-approximation algorithm based on a certain LP relaxation of the problem.
The approximation factor was improved to $O(\sqrt{\log n})$ in another seminal paper by Arora, Rao and Vazirani~\cite{ARV}
who used an SDP relaxation. Arora, Hazan and Kale~\cite{AHK}
showed how to (approximately) solve this SDP in $\tilde O(n^2)$ time using multicommodity flows
while preserving the $O(\sqrt{\log n})$ approximation factor.
Arora and Kale later developed in~\cite{AK} a more general method for solving SDPs
that allowed different tradeoffs between approximation factor and complexity;
in particular, they presented an $O(\log n)$-approximation algorithm using $O(\log^{O(1)} n)$ maxflow computations,
and a simpler version of $O(\sqrt{\log n})$-approximation with $\tilde O(n^2)$ complexity.

The algorithms in~\cite{AHK,AK} were based on the Multiplicative Weights Update method.
An alternative approach based on the so-called {\em cut-matching game} was proposed by 
Khandekar, Rao and Vazirani~\cite{Khandekar:06}; their method computes $O(\log^2 n)$-approximation for {\sc Sparsest Cut}
using $O(\log^{O(1)} n)$ maxflows. This was later improved to $O(\log n)$-approximation by
Orecchia, Schulman,  Vazirani and Vishnoi~\cite{Orecchia:08}.

%\myparagraph{Sherman's algorithm}
The line of works above culminated in the result of Sherman~\cite{Sherman}, who showed
how to compute $O(\sqrt{(\log n)/\varepsilon})$-approximation for {\sc Sparsest Cut} 
using $O(n^\varepsilon\log^{O(1)}n)$ maxflows for any $\varepsilon\in[\Theta(1/\log n),\Theta(1)]$.
This effectively subsumes previous results, as taking $\varepsilon=\Theta(1/\log n)$ yields
an $O(\log n)$ approximation using $O(\log^{O(1)} n)$ maxflows,
while a sufficiently small constant $\varepsilon$ achieves an $O(\sqrt{\log n})$-approximation
and improves on the $\tilde O(n^2)$ runtime in~\cite{AHK,AK}.
In particular, using the recent almost linear-time maxflow algorithm~\cite{Chen:FOCS22}
yields $O(n^{1+\varepsilon})$ complexity
for $O(\sqrt{\log n})$-approximation.
(As usual, we assume in this paper that graph $G$ has $m=O(n\log n)$ edges. 
This can be achieved by  sparsifying the graph using the algorithm of Bencz\'ur and Karger \cite{BenczurKarger:96},
which with high probability preserves the cost of all cuts up to any given constant factor.)

\myparagraph{Our contributions} 
In this paper we present a new algorithm that computes 
an $O(\sqrt{(\log n)/\varepsilon})$ approximation for {\sc Sparsest Cut} w.h.p.\ whose expected runtime is $O((n^\varepsilon\log^{O(1)}n)\cdot T_{\tt maxflow})$ for given $\varepsilon\in[\Theta(1/\log n),\Theta(1)]$,
where $T_{\tt maxflow}=\Omega(n)$ is the runtime of a maxflow algorithm on a graph with $n$ nodes and $O(n\log n)$ edges.
It has the following features: \\
(i) It simplifies Sherman's algorithm in two different ways. \\
(ii) The algorithm is parallelizable: it can be implemented on $O(n^\varepsilon)$ processors in expected parallel runtime $O((\log^{O(1)}n)\cdot T_{\tt maxflow})$ 
(in any version of the PRAM model), where ``parallel runtime'' is defined as the maximum runtime over all processors.
%Note that this is an exponential improvement over complexity $O((n^\varepsilon\log^{O(1)}n)\cdot T_{\tt maxflow})$ of the sequential version. 
\\
(iii) To prove algorithm's correctness, we introduce a new technique, which we believe may yield
smaller constants in the $O(\cdot)$ notation. 
Note that there are numerous papers that optimize constants for problems with a constant factor
approximation guarantee. 
We argue that this direction makes just as much sense for the {\sc Sparsest Cut} problem.
The question can be naturally formulated as follows: 
what is the fastest algorithm to compute $C\sqrt{\log n}$-approximation for a given constant $C$?
Alternatively, for maxflow-based algorithms one may ask
 what is the smallest $C=C_\varepsilon$ such that there is a $C\sqrt{\log n}$-approximation algorithm that uses $\tilde O(n^\varepsilon)$ maxflow,
for given $\varepsilon>0$. 
Unfortunately, it is impossible to directly compare our constant with that of Sherman:
we believe that the paper~\cite{Sherman} contains a numerical mistake
(see the footnote in Section~\ref{sec:chaining}).
An additional complication is that optimizing the constants may not be an easy task. 
Due to these considerations, we formulate our claim differently:
our proof technique should lead to smaller constants since our analysis is more compact and avoids case analysis present in~\cite{Sherman}.

To explain details, we need to give some background on Sherman's algorithm.
It builds on the work of Arora and Kale~\cite{AK} who approximately solve an SDP relaxation
using a (matrix) Multiplicative Weights update algorithm (MW).
The key subroutine is to identify constraints that are violated by the current primal solution.
Both~\cite{AK} and~\cite{Sherman} do this by solving a multicommodity flow problem.
Arora and Kale simply call Fleischer's multicommodity flow algorithm~\cite{Fleischer:00},
while Sherman designs a more efficient customized method for approximately solving this flow problem
using another application of MW. Our algorithm avoids solving the multicommodity flow problem,
and instead searches for ``violating paths'', i.e.\ paths that violate triangle inequalities in the SDP relaxation.
We show that this can be done by a simple randomized procedure that does not rely on MW.
Furthermore, independent calls to this procedure return violating paths that are mostly disjoint,
which allows parallelization: we can compute many such paths on different processors and
then take their union.

In order to compute violating paths, we first design procedure ${\tt Matching}(u)$
that takes vector $u\in\mathbb R^d$ and outputs a directed matching on a given set $S\subset\mathbb R^d$
with $|S|=\Theta(n)$.
(It works by calling a maxflow algorithm and then postprocessing the flow).
The problem then boils down to the following:
given vector $u$ randomly sampled from the Gaussian distribution, we need to sample vectors $u_1,\ldots,u_K$
so that set ${\tt Matching}(u_1)\circ\ldots\circ{\tt Matching}(u_K)$ contains many paths $(x_0,x_1,\ldots,x_K)$ 
with  ``large stretch'', i.e.\ with $\langle x_K-x_0,u\rangle\ge \Omega(K)$.
(Here ``$\circ$'' is the operation that ``chains together'' paths in a natural way).
This was also a key task in~\cite{Sherman}, where it was needed for implementing one inner MW step.
We refer to an algorithm that samples vectors $u_1,\ldots,u_K$ as a {\em chaining algorithm}.

Sherman's chaining algorithm can be viewed as an algorithmization of the proof in the original ARV paper~\cite{ARV}
and its subsequent improvement by Lee~\cite{Lee:05}. We present a simpler chaining algorithm with a very different proof;
as stated before, the new proof may yield smaller constants.

\myparagraph{Concurrent work} After finishing the first draft of the paper~\cite{vnk-arxiv}, we learned
about a very recent work by Lau-Tung-Wang~\cite{LTW},
which considers a generalization of {\sc Sparsest Cut} to {\em directed} graphs.
The authors presented an algorithm which, in the case of undirected graphs,
also simplifies Sherman's algorithm by computing violating paths instead of solving a multicommodity flow problem.
Unlike our paper,~\cite{LTW} does not consider parallelization, and uses Sherman's chaining algorithm as a black box.

Another very recent paper by Agarwal-Khanna-Li-Patil-Wang-White-Zhong~\cite{AKLPWWZ:24}
presented a parallel algorithm for {\em approximate maxflow}  with polylogarithmic depth and near-linear work in the PRAM model.
Using this algorithm, they presented, in particular, an $O(\log^3 n)$-approximation sparsest cut algorithm with polylogarithmic depth and near-linear work
(by building on the work~\cite{Nanongkai:STOC17} who showed how to replace exact maxflow computations in the cut-matching game~\cite{Khandekar:06,Orecchia:08}
with approximate maxflows).

\myparagraph{Other related work} The problem of computing a cut of conductance $\tilde O(\sqrt\phi)$
assuming the existence of a cut of conductance $\phi\in[0,1]$ (with various conditions on the balancedness)
has been considered in~\cite{Sarma:15,Kuhn:15,Orecchia:SODA11}.
Papers~\cite{Sarma:15,Kuhn:15} presented distributed algorithms for this problem (in the CONGEST model),
while~\cite{Peng:STOC14} observed that the {\sc BalCut} algorithm in~\cite{Orecchia:SODA11}
is parallelizable: it can be implemented in near-linear work and polylogarithmic depth.
Note that $\phi$ can be much smaller than 1.
\cite{Chang:PODC19} presented a distributed CONGEST algorithm for computing an expander decomposition of a graph.

\section{Background: Arora-Kale framework}
We will describe the algorithm only for the {\sc $c$-Balanced Separator} problem.
As shown in~\cite{AK}, the {\sc Sparsest Cut} problem can be solved by a very similar approach,
essentially by reducing it to the {\sc $c$-Balanced Separator} problem for some constant $c$;
we refer to~\cite{AK} for details.\footnote{
These details are actually not included in the journal version~\cite{AK}, but can be found
in~\cite{Kale:PhD} or in \url{https://www.cs.princeton.edu/~arora/pubs/mmw.pdf}, Appendix A.
}

Let $c_e\ge 0$ be the weight of edge $e$ in $G$.
The standard SDP relaxation of the {\sc $c$-Balanced Separator} problem can be written in vector form as follows (see~\cite{ARV}):
\begin{subequations}\label{eq:SDP:orig}
\begin{align}
\min \sum_{e=\{i,j\}\in E} & c_e || v_i - v_j||^2 \label{eq:SDP:orig:a} \\
||v_i||^2&=1                                       && \forall i \label{eq:SDP:orig:b} \\
||v_i-v_j||^2+||v_j-v_k||^2&\ge || v_i-v_k||^2     && \forall i,j,k \label{eq:SDP:orig:c} \\
\sum_{i<j}||v_i-v_j||^2 &\ge 4c(1-c)n^2 \label{eq:SDP:orig:d}
\end{align}
\end{subequations}

Here $v_i\in\mathbb R^n$ for each node $i\in V$. The optimum of this SDP divided by 4 is a lower bound on the minimum {\tt $c$-Balanced Separator} problem.
Arora and Kale considered a slightly different relaxation:

\begin{subequations}\label{eq:SDP}
\begin{align}
\min \sum_{e=\{i,j\}\in E} & c_e || v_i - v_j||^2  &  \min C &\bullet X \label{eq:SDP:a} \\
||v_i||^2&=1                                                                               &X_{ii}&=1                           && \forall i \label{eq:SDP:b} \\
\sum_{j=1}^{\ell(p)} ||v_{p_j}-v_{p_{j-1}}||^2 & \ge || v_{p_{\ell(p)}}-v_{p_0}||^2           &T_p\bullet X &\ge 0                   && \forall p \label{eq:SDP:c} \\
\sum_{i,j\in S:i<j}||v_i-v_j||^2 &\ge \xi n^2                                                     &K_S\bullet X&\ge \xi n^2                  && \forall S \label{eq:SDP:d} \\
                             &                                                             &X&\succeq  0 \label{eq:SDP:e}
\end{align}
\end{subequations}
Here $p$ stands for a path $p=(p_0,\ldots,p_{\ell(p)})$ in graph $G$, notation ``$\forall S$'' means all subsets $S\subseteq V$ 
of size at least $(1-c/4)n$, and $\xi=3c-4c^2$.
Matrix $X$ is defined via $X=V^T V$ where $V$ is the $n\times n$ matrix with columns $v_1,\ldots,v_n$.
We have  $C=\calL(G)$
and $T_p=\calL(G_p)-\calL(G_{(p_0,p_{\ell(p)})})$
where $\calL(\cdot)$ is the Laplacian of the corresponding graph and $G_q$ for a path $q$ is the undirected unweighted graph containing all edges of $q$.
Finally, $K_S(i,i)=|S|-1$ for $i\in S$,
$K_S(i,j)=-1$ for distinct $i,j\in S$,
and $K_S(i,j)=0$ in all other cases.

Note that triangle inequalities~\eqref{eq:SDP:orig:c} imply path inequalities~\eqref{eq:SDP:c},
while constraints~\eqref{eq:SDP:orig:b} and~\eqref{eq:SDP:orig:d} imply constraints~\eqref{eq:SDP:d} (see~\cite{AK}).
SDP~\eqref{eq:SDP} may be looser than~\eqref{eq:SDP:orig},
but its optimum divided by 4 is still a lower bound on the minimum {\tt $c$-Balanced Separator} problem.

In the sequel we will use a slight modification of~\eqref{eq:SDP}
in which constraints~\eqref{eq:SDP:c} are enforced for {\em all} sequences $p=(p_0,\ldots,p_{\ell(p)})$ 
of distinct nodes  in $G$, and not just paths in $G$.
This modification can only make the relaxation stronger;
constraints~\eqref{eq:SDP:c} are now equivalent to the triangle inequalities in~\eqref{eq:SDP:orig:c}.

The dual of~\eqref{eq:SDP} is as follows. It has variables $y_i$ for every node $i$,
$f_p$ for every path $p$, and $z_S$ for every set $S$ of size at least $(1-c/4)n$.
Let ${\tt diag}(y)$ be the diagonal matrix with vector $y$ on the diagonal:
\begin{subequations}\label{eq:SDPdual}
\begin{align}
\max \sum_i y_i + \xi n^2 \sum_S z_S \\
{\tt diag}(y) + \sum_p f_p T_p + \sum_S z_S K_S \preceq C \\
f_p,z_S \ge 0 \quad \forall p,S
\end{align}
\end{subequations}

\subsection{Matrix multiplicative weights algorithm}\label{sec:MW}
To solve the above SDP, \cite{AK} converts it to a sequence of feasibility problems.
First, an interval $[L,U]$ is computed containing an optimal value of the objective.
One can use, for example, the algorithm in~\cite{Orecchia:08} to get such an interval
with $U/L=O(\log n)$ via $O(\log^{O(1)} n)$ maxflows.
Fix a value $\alpha\in[L,U]$, and replace the objective \eqref{eq:SDP:a} with the constraint
\begin{equation}
 C\bullet X  \le \alpha \label{eq:SDP:a'}  \tag{\ref{eq:SDP:a}$'$}
\end{equation}
It suffices to try $O(\log (U/L))=O(\log \log n)$ values of threshold $\alpha$ if we are willing
to accept the loss by a constant factor in the approximation ratio.

Let (\ref{eq:SDP}$'$) be the system consisting of constraints \eqref{eq:SDP:a'} and \eqref{eq:SDP:b}-\eqref{eq:SDP:e}.
To check the feasibility of this system, Arora and Kale apply the {Matrix Multiplicative Weights} (MW) algorithm which
we review in Appendix~\ref{sec:app:MW}. The main computational subroutine is procedure {\tt Oracle}
that, given current matrix $X$ of the form $X=V^T V$, $V\in\mathbb R^{n\times n}$, should either (i) find an inequality violated by $X$,
or (ii) find a $\Theta(1)$-balanced cut of value at most $\kappa\alpha$ where $\Theta(\kappa)$ is the desired
approximation factor. 
Working directly with matrix $V$ would be too slow (even storing it requires $\Theta(n^2)$ space and thus $\Omega(n^2)$ time).
To reduce complexity, Arora and Kale work instead with matrix $\tilde V\in\mathbb R^{d\times n}$,
$d\ll n$ so that $X\approx \tilde X\eqdef \tilde V^T \tilde V$.
Let $v_1,\ldots,v_n\in\mathbb R^n$ and $\tilde v_1,\ldots,\tilde v_n\in\mathbb R^d$
be the columns of $V$ and $\tilde V$, respectively. Below we give a formal specification of {\tt Oracle}.
Note that the oracle has access only to vectors $\tilde v_1,\ldots,\tilde v_n$.

\vspace{5pt}
\noindent\hspace{0pt}
\begin{minipage}{\dimexpr\columnwidth-10pt\relax}
\fbox{\parbox{\textwidth}{

\noindent {\bf Input}:
 vectors $v_1,\ldots,v_n\in\mathbb R^n$ and $\tilde v_1,\ldots,\tilde v_n\in\mathbb R^d$
satisfying 
\begin{subequations}
\begin{align}
||\tilde v_i||^2&\;\;\le\;\; 2 && \forall i \label{eq:tilde-v:one} \\
\sum_{i,j\in V:i<j}||\tilde v_i-\tilde v_j||^2&\;\;\ge\;\; \tfrac {\xi n^2}4 \label{eq:tilde-v:two} \\
|\;||\tilde v_i||^2-||v_i||^2\;| &\;\;\le\;\; \gamma (||\tilde v_i||^2+\tau) && \forall i \label{eq:GramApproximation:i} \\
|\;||\tilde v_i-\tilde v_j||^2-||v_i-v_j||^2\;| &\;\;\le\;\; \gamma (||\tilde v_i-\tilde v_j||^2+\tau) && \forall i,j \label{eq:GramApproximation:ij}
\end{align}
\end{subequations}
for some constants $\gamma,\tau>0$. Let $X=V^T V$ and $\tilde X=\tilde V^T\tilde V$
 where $V\in\mathbb R^{n\times n}$ and $\tilde V\in\mathbb R^{d\times n}$
 are the matrices with columns $\{v_i\}$ and $\{\tilde v_i\}$, respectively.

\noindent {\bf Output}: either 
(i) variables $f_p\ge 0$ and symmetric matrix $F\preceq C$ such that 
\begin{equation}
\left(\sum\nolimits_p f_p T_p - F\right)\bullet X\le -\alpha\label{eq:OracleInequality}
\end{equation}
 or 
(ii) a $\Theta(1)$-balanced cut of value at most $\kappa\alpha$.

}}
\end{minipage}
\vspace{5pt}

The number of iterations of the MW algorithm will depend on the maximum possible spectral norm of matrix $\tfrac \alpha n I + \sum_p f_p T_p - F$.
This parameter is called the {\em width} of the oracle, and will be denoted as $\rho$.
We will use the bound $\rho=||\tfrac \alpha n I + \sum_p f_p T_p - F||\le \tfrac \alpha n + ||\sum_p f_p T_p - F||$.
The algorithm has the following guarantee (see Appendix~\ref{sec:app:MW} for details).\!\!\!

\begin{theorem}[\cite{AK}]
If in the MW algorithm the first $T=\lceil \frac{4\rho^2 n^2 \ln n}{\epsilon^2}\rceil$
calls to {\tt Oracle} output 
 option {\em (i)} then the optimum value of SDP~\eqref{eq:SDP} is at least $\alpha-\epsilon$.
\end{theorem}

\subsection{Oracle implementation}

To implement the oracle, Arora and Kale interpret values $f_p$ as a {\em multicommodity flow}
in graph $G$, i.e.\ a flow that sends $f_p$ units of demand between the endpoints of $p$.
Given this flow, introduce the following notation. Let $f_e$ be the flow on edge $e$ (i.e.\ $f_e=\sum_{p \ni e}f_p$).
Let $d_{ij}$ be the total flow between nodes $i$ and $j$ 
(i.e.\ $d_{ij}=\sum_{p\in\calP_{ij}}f_p$ 
where $\calP_{ij}$ 
is the set of paths from $i$ to $j$).
Finally, let $d_i$ be the total flow from node $i$ (i.e.\ $d_i=\sum_{j}d_{ij})$.
Given parameter $\pi>0$, a {\em valid $\pi$-regular flow} is one that satisfies capacity constraints:
$f_e\le c_e$ for all edges $e$ and $d_i\le \pi$ for all nodes $i$.

The oracle in~\cite{AK} computes a $\pi$-regular flow $f$ for some parameter $\pi$,
and sets $F$ to be the Laplacian of the {\em flow graph} (i.e.\ the weighted graph where edge $e$ has weight $f_e$).
Capacity constraints then ensure that $F\preceq C$ (because $C-F$ is the Laplacian of the weighted graph with weights $c_e-f_e\ge 0$
on edges $e\in E$).
Let $D$ be the Laplacian of the {\em demand graph} (i.e.\ the complete weighted graph where edge $\{i,j\}$ has weight $d_{ij}$).
It can be checked that $\sum_p f_p T_p=F-D$.
Thus, the oracle needs to ensure that $D\bullet X\ge\alpha$, or equivalently
\begin{equation}\label{eq:largeD}
\sum_{i<j} d_{ij} ||v_i-v_j||^2 \ge \alpha
\end{equation}
All degrees in the demand graph are bounded by $\pi$, therefore $||D||\le 2\pi$.
Thus, the width of the oracle can be bounded as 
$||\rho||\le \tfrac \alpha n + ||D||\le \tfrac \alpha n + 2\pi$.

%for path $p=(p_1,\ldots,p_{\ell(p)})$ as a flow in graph $G$ between $p_1$ and $p_\ell$.
%They set $F$

Below we summarize three known implementations of the oracle.
The first two are due to Arora and Kale~\cite{AK} and the third one is due to Sherman~\cite{Sherman}.
%In all three cases $c'$ is some constant in $(0,c)$ (that depends on $c$).
%Different implementations of the oracle are described in~\cite{AK} and \cite{Sherman}. They have the following properties.
\begin{enumerate}
\item Using $O(1)$ expected maxflow computations, the oracle computes either a $\pi$-regular flow with $\pi=O(\tfrac {\alpha\log n}n)$
or a $\Theta(1)$-balanced cut of capacity at most $O(\alpha\log n)$.
\item Using $O(1)$ expected multicommodity flow computations, the oracle computes either a $\pi$-regular flow with $\pi=O(\tfrac {\alpha}n)$
or a $\Theta(1)$-balanced cut of capacity at most $O(\alpha\sqrt{\log n})$.
\item Let $\varepsilon\in[O(1/\log n),\Omega(1)]$. Using $O(n^\varepsilon \log^{O(1)}n)$ expected maxflow computations,
the oracle either computes a $\pi$-regular flow with $\pi=O(\tfrac {\alpha}{\varepsilon n})$,
or a $\Theta(1)$-balanced cut of capacity at most $O\left(\alpha\sqrt{\tfrac{\log n}\varepsilon}\right)$.
\end{enumerate}
By the discussion in Section~\ref{sec:MW}, these oracles lead to algorithms with approximation factors $O(\log n)$, $O(\sqrt{\log n})$ and $O\left(\sqrt{\tfrac{\log n}\varepsilon}\right)$, respectively.

To conclude this section, we discuss how to verify condition~\eqref{eq:largeD} in practice.
(Recall that we only have an access to approximations $\tilde v_i$ of vectors $v_i$.)
%We will use the assumption that flows $f_p$ returned by the oracle satisfy condition $d_i=0$ for all $i\in V-S$.
%(Note that oracles used in~\cite{AK,Sherman} do satisfy this).
%Under this assumption, we can use the following result.
\begin{proposition}\label{prop:NGANGKAS}
Suppose parameters $\tau,\gamma$ in eq.~\eqref{eq:GramApproximation:i}-\eqref{eq:GramApproximation:ij} satisfy $\tau\le 2$ and $\gamma\le \tfrac{\alpha}{20n\pi}$. Then condition
\begin{equation}\label{eq:largeDtilde}
\sum_{i<j} d_{ij} ||\tilde v_i-\tilde v_j||^2 \ge 2\alpha
\end{equation}
implies condition~\eqref{eq:largeD}.
\end{proposition}
\begin{proof}
Denote $z_{ij}=||v_i-v_j||^2$ and $\tilde z_{ij}=||\tilde v_i-\tilde v_j||^2$.
We have $||\tilde v_i||^2\le 2$, $||\tilde v_j||^2\le 2$ and hence $\tilde z_{ij}\le 8$.
By Theorem~\ref{th:GramApproximation} we then have $|\tilde z_{ij}-z_{ij}|\le\gamma(\tilde z_{ij}+\tau)\le 10\gamma$.
This implies that 
 $\sum_{i<j} d_{ij} |\tilde z_{ij}-z_{ij}| \le 10\gamma \cdot \sum_{i<j}d_{ij}\le 10\gamma \cdot 2n\pi\le \alpha$.
 The claim follows.
\end{proof}

%%%%%%%%%%%%%%%%%%%%%%%%%%%%%%%%%%%%%%%%%%%%%%%%%%%%%%%%%%%%%%%%%%%%%%%%%%%%%%%%%%%%%%%%%%%%%%%%%%%%%%%%%%%%%%%%%%%%%%%%%

\section{Our algorithm}
In this section we present our implementation of the oracle.
To simplify notation, we assume in this section that vectors $\tilde v_i$ for $i\in V$ are unique,
and rename the nodes in $V$ so that $\tilde v_x=x$ for each $x\in V$.
Thus, we now have $V\subseteq\mathbb R^d$. The ``true'' vector in $\mathbb R^n$ corresponding to $x\in V$ is still denoted as $v_x$.

Recall that the oracles in~\cite{AK,Sherman} do one of the following: 
{
\setlist{nolistsep}\em
\begin{itemize}
\item  output a cut; 
\item output multicommodity flows $f_p$ satisfying~\eqref{eq:largeD}, and set $F$ to be its flow graph. 
\end{itemize}\em
\noindent Our oracle will use a third option described in the lemma below.
\begin{lemma}\label{lemma:ViolatingPaths}
Let $M$ be a set of paths on $V$ such that each $p\in M$ violates the path inequality by some amount $\tfrac 12\Delta>0$.
In other words, we require that $T_p\bullet X\le -\tfrac 12\Delta$, or equivalently
\begin{equation}\label{eq:PathViolation}
\sum_{j=1}^{\ell(p)} ||v_{p_j}-v_{p_{j-1}}||^2  \;\;\le\;\; || v_{p_{\ell(p)}}-v_{p_0}||^2 - \tfrac 12\Delta
\end{equation}
Let $\calG_F$ and $\calG_D$ be respectively flow and demand graphs of the multicommodity flow
defined by $M$ (where each path carries one unit of flow).
Set $f_p=\frac{2\alpha}{|M|\Delta}$ for all $p\in M$, $f_p=0$ for $p\notin M$, and $F=0$.
Then these variables give a valid output of the oracle with width 
$\rho\le \tfrac \alpha n + \tfrac {4\alpha(\pi_F+\pi_D)}{|M|\Delta}$
where $\pi_F,\pi_D$ are the maximum degrees of $\calG_F,\calG_D$, respectively.
\end{lemma}
\begin{proof}
We have $F\preceq C$ and $(\sum_p T_p-F)\bullet X\le\sum_p f_p  \cdot(-\tfrac 12\Delta)=-\alpha$,
so condition~\eqref{eq:OracleInequality} holds.
It can be checked that $\sum_p f_p T_p = \frac{2\alpha}{|M|\Delta}(\tilde F-\tilde D)$
where $\tilde F$ and $\tilde D$ are the Laplacians of respectively $\calG_F$ and $\calG_D$.
Therefore,
$\rho=||\tfrac \alpha n I + \sum_p f_p T_p - F|| 
=||\tfrac \alpha n I + \tfrac{2\alpha}{|M|\Delta}(\tilde F-\tilde D)|| \
\le ||\tfrac \alpha n I|| + \tfrac{2\alpha}{|M|\Delta}(||\tilde F||+||\tilde D||)
\le\tfrac \alpha n +\tfrac{2\alpha}{|M|\Delta}(2\pi _F+2\pi_ D)$.
\end{proof}
Recall that in the specification of the oracle we required paths $p$ to have distinct nodes
(to make  the number of constraints finite). This does not cause problems
for Lemma~\ref{lemma:ViolatingPaths}: if some path in $M$ does not satisfy this,
then before applying the lemma we can shorten it while preserving endpoints and condition~\eqref{eq:PathViolation}.
%Note that we do not require paths in $M$ to be paths in the original graph $G$.
%Thus, we need to use SDP relaxation~\eqref{eq:SDP} with all possible paths $p$ (with a certain bound on the length),
%which is still a valid relaxation of the problem.

The following proposition shows how to verify condition~\eqref{eq:PathViolation}
for unobserved variables $v_x\in\mathbb R^n$ using observed variables $x\in\mathbb R^d$.
Its proof is very similar to that of Proposition~\ref{prop:NGANGKAS}, and is omitted.
\begin{proposition}
Suppose parameters $\tau,\gamma$ in Theorem~\ref{th:GramApproximation} satisfy $\tau\le 2$ and $\gamma\le \tfrac{\Delta}{20(K+1)}$ for some integer $K\ge 1$. Then condition
\begin{equation}\label{eq:PathViolationtilde}
\sum_{j=1}^{\ell(p)} ||p_j-p_{j-1}||^2  \;\;\le\;\; || p_{\ell(p)}-p_0||^2 - \Delta
\end{equation}
 implies condition~\eqref{eq:PathViolation}, assuming that $\ell(p)\le K$.
\end{proposition}
%\begin{proof}
%Consider nodes $x,y\in S$,
%and denote Denote $\tilde z=||y-x||^2$ and $z=||v_y-v_x||^2$.
% Conditions~\eqref{eq:GramApproximation} give
%$|z-\tilde z|\le \gamma(8+\tau)\le \Delta/(2(k+1))$.
%\end{proof}

\begin{remark}
Sherman~\cite{Sherman} explicitly tries to (approximately) solve the multicommodity flow problem:
find valid flows $\{f_p\}$ in $G$ with demands $\{d_{xy}\}$ that maximize $\sum_{x,y}d_{xy}||x-y||^2$.
This is done by an iterative scheme via the Multiplicative Weights (MW) framework.
Using the option in Lemma~\ref{lemma:ViolatingPaths} has the following advantages over this approach.
\begin{enumerate}
\item We can avoid another application of MW and thus simplify the algorithm.
\item The oracle can be easily parallelized: we can compute different ``violating paths'' on different processors
and then take their union. (Of course, we still need to make sure that these paths are ``sufficiently disjoint''
so that the degrees of graphs $\calG_F,\calG_D$ and hence the oracle width remain small). 

\end{enumerate}
\end{remark}

\subsection{Correlated Gaussians and measure concentration}

We  write $u\sim\calN$ to indicate that $u$ is a random vector in $\mathbb R^d$ with Gaussian independent components $u_i\sim\calN(0,1)$.
Throughout the paper notation ${\tt Pr}_u[\cdot]$ means the probability under distribution $u\sim\calN$.
We write $(u,u')\sim\calN_\omega$ for $\omega\in[0,1)$ to indicate that $(u,u')$ are random vectors in $\mathbb R^d\times \mathbb R^d$
such that for each $i\in[d]$, pair $(u_i,u'_i)$ is an independent 2-dimensional Gaussian with mean $(0,0)^T$ and covariance matrix $\begin{pmatrix}1 & \omega \\ \omega & 1\end{pmatrix}$.
We write $u'\sim_\omega u$ to indicate that  $u'$ is an {\em $\omega$-correlated copy of~$u$}~\cite{Sherman},
i.e.\ $(u,u')$ is generated according to $(u, u')\sim\calN_\omega$ conditioned on fixed $u$.
It can be checked that for each $i\in [d]$, $u'_i$ is an independent Gaussian with mean $\omega \cdot u_i$ and variance $1-\omega^2$.
Note, if $(u,u')\sim\calN_\omega$ then $u\sim\calN$ and $u'\sim\calN$.
Conversely, the process $u\sim\calN,u'\sim_\omega u$ generates  pair $(u,u')$ with distribution $\calN_\omega$.
The same is true for the process  $u'\sim\calN, u\sim_\omega u'$.

The key property for obtaining an $O(\sqrt{\log n})$-approximation algorithm is {\em measure concentration} of the Gaussian distribution.
This property can be expressed in a number of different ways; we will use the following version.\footnote{Theorem~\ref{th:isoperimetric}
is formulated in \cite{Mossel06} for the discrete cube, but the proof also works for the Gaussian distribution.
For completeness, we reproduce the proof in Appendix~\ref{sec:isoperimetric}. %~\ref{sec:isoperimetric}.
}
\begin{theorem}[\cite{Mossel06}]\label{th:isoperimetric}
Consider sets $\calA\subseteq\mathbb R^d$, $\calB\subseteq\mathbb R^d$ 
with ${\tt Pr}_{u}[u\in \calA]={\tt Pr}_{u'}[u'\in \calB]=\delta$.
Then
$$
{\tt Pr}_{(u,u')\sim\calN_\omega}[(u,u')\in \calA\times \calB]\ge \delta^{2/(1-\omega)}
$$
\end{theorem}

Given a sequence of numbers $\omega_1,\ldots,\omega_{k-1}\in[0,1)$,
we write $(u_1,\ldots,u_k)\sim\calN_{\omega_1,\ldots,\omega_{k-1}}$
to indicate the following distribution:
sample $u_1\sim\calN$, then $u_2\sim_{\omega_1} u_1$,  then $u_3\sim_{\omega_2} u_2$, $\ldots$,  then $u_k\sim_{\omega_{k-1}} u_{k-1}$.
If $\omega_1\!=\!\ldots\!=\!\omega_{k-1}\!=\!\omega$ then we write $\calN^{k}_\omega$ instead of $\calN_{\omega_1,\ldots,\omega_{k-1}}$ for brevity.
Finally, if some values of the sequence $(u_1,\ldots,u_k)$
are fixed, e.g.\ $(u_1,u_k)$, then we write
$(u_1,\ldots,u_k)\sim\calN_{\omega_1,\ldots,\omega_{k-1}}|(u_1,u_k)$
to indicate that $(u_1,\ldots,u_k)$ is obtained by sampling from $\calN_{\omega_1,\ldots,\omega_{k-1}}$
conditioned on fixed values $(u_1,u_k)$.
In that case $(u_2,\ldots,u_{k-1})$ are random variables that depend on $(u_1,u_k)$.

\subsection{Procedure ${\tt Matching}(u)$}
In this section we describe a procedure that takes vector $u\in\mathbb R^d$ and either outputs a directed matching $M$ on nodes $V$
or terminates the oracle. In this procedure we choose constants $c',\Delta,\sigma$ (to be specified later), and denote
\begin{eqnarray}
%\Delta&=&\sqrt{\frac\varepsilon{\log n}} \\
\pi&=&\frac{6 \alpha }{c' n\Delta}
\end{eqnarray}
%where  $c',\Delta$ are positive parameters that will be specified later.

\begin{algorithm}[H]
  \DontPrintSemicolon
\SetNoFillComment

	compute $w_x=\langle x,u\rangle$ for each $x\in V$ \\
	sort $\{w_x\}_{x\in V}$, let $A,B$ be subsets of $V$ with $|A|=|B|=2c' n$ containing nodes with the least and the greatest values of $w_x$, respectively \\
%	if exists $x\in A$, $y\in B$ with $w_y-w_x<\sigma$ then return empty matching $M=\varnothing$ \\
%	if $\min_{x\in A,y\in B}(w_y-w_x)<\sigma$ then return empty matching \\
	let $G'$ be the graph obtained from $G$ by adding new vertices $s,t$ and edges $\{\{s,x\}\::\:x\in A\}\cup \{\{y,t\}\::\:y\in B\}$ of capacity $\pi$ \\
	compute maximum $s$-$t$ flow and the corresponding minimum $s$-$t$ cut in $G'$ \\
	if capacity of the cut is less than $c'n\pi=\tfrac {6\alpha}\Delta$ then return this cut and terminate the oracle \\
	use flow decomposition to compute multicommodity flows $f_p$ and demands $d_{xy}$ (see text) \\
	if flows $f_p$ satisfy condition~\eqref{eq:largeDtilde} then return these flows and terminate the oracle \\
%	let $M^\ast=\{(x,y)\::\:x\in A,y\in B,\exists p=(x,\ldots,y)$ s.t.\ $f_p>0,w_y-w_x\ge \sigma,||x-y||^2\le \Delta\}$ \\
	let $M_{\tt all}=\{(x,y)\in A\times B\::\:d_{xy}>0,w_y-w_x\ge \sigma\}$ and $M_{\tt short}=\{(x,y)\in M_{\tt all}\::\:||x-y||^2\le \Delta\}$ \\
	pick maximal matching $M\subseteq M_{\tt short}$ and return $M$
      \caption{${\tt Matching}(u)$. 
      }\label{alg:Matching}
\end{algorithm}

Let us elaborate line 6. Given flow $f'$ in $G'$, we compute its flow decomposition and remove flow cycles.
Each path in this decomposition has the form $p'=(s,p,t)$ where $p=(x,\ldots,y)$ with $x\in A$, $y\in B$.
For each such $p$ we set $f_p=f'_{p'}$, and accordingly increase demand $d_{xy}$ by $f'_{p'}$.
Note that we need to know only the endpoints of $p$, and not $p$ itself.
This computation can be done in $O(m \log n)$ time using dynamic trees~\cite{SleatorTarjan}.
(The same subroutine was used in~\cite{Sherman}).

For the purpose of analysis we make the following assumption: if Algorithm~\ref{alg:Matching} terminates
at line 5 or 7 then it returns $\varnothing$ (the empty matching).
Thus, we always have ${\tt Matching}(u)\subseteq V\times V$ and $|{\tt Matching}(u)|\le |V|$.
%The following result is proved in Appendix~\ref{sec:lemma:matching}.
\begin{lemma}\label{lemma:matching}
(a) If the algorithm terminates at line 5 then the returned cut is $c'$-balanced. \\
(b) There exist positive constants $c',\sigma,\delta$ for which  either (i) $\mathbb E_{u} |{\tt Matching}(u)|\ge \delta n$,
or (ii)~Algorithm~\ref{alg:Matching} for $u\sim\calN$ terminates at line 5 or 7 with probability at least $\Theta(1)$.
\end{lemma}
\begin{proof}
\myparagraph{(a)} We repeat the argument from~\cite{AK}. Since the total flow is at most $\pi c'n$, at least $c'n$ newly added source edges are not saturated by the flow,
and hence their endpoints are on the side of the source node in the cut obtained, which implies that that side has at least $c'n$ nodes.
Similarly, the other side of the cut also has at least $c'n$ nodes, and thus the cut is $c'$-balanced.

\myparagraph{(b)} 
Assume that condition (ii) is false.
By a standard argument, conditions~\eqref{eq:tilde-v:one} and \eqref{eq:tilde-v:two} imply the following:
there exist constants $c'\in(0,c)$, $\beta>0$ and $\sigma>0$ such that with probability at least $\beta$ Algorithm~\ref{alg:Matching} reaches line 9 and we have $w_y-w_x\ge \sigma$ for all $x\in A$, $y\in B$
(see~\cite[Lemma 14]{Kale:PhD}). Suppose that this event happens.
We claim that in this case $|M|\ge \tfrac 13 c'n$.
Indeed, suppose this is false. Let $A'\subseteq A$ and $B'\subseteq B$ be the sets of nodes involved in $M$
(with $|A'|=|B'|=|M|=k$). The total value of flow from $A$ to $B$ is at least $c'\pi n$ (otherwise we would have terminated at line 5).
The value of flow leaving $A'$ is at most $|A'|\cdot \pi\le \tfrac 13 c'\pi n$.
Similarly, the value of flow entering $B'$ is at most $|B'|\cdot \pi\le \tfrac 13 c'\pi n$.
Therefore, the value of flow from $A-A'$ to $B-B'$ is at least $c'\pi n-2\cdot \tfrac 13 c'\pi n=\tfrac 13 c'\pi n$.
For each edge $(x,y)\in M_{\tt all}$ with $x\in A-A'$, $y\in B-B'$ we have $||x-y||^2 > \Delta$
(otherwise $M$ would not be a maximal matching in $M_{\tt short}$). Therefore,
$$
\sum_{p:p=(x,\ldots,y)} f_p ||x-y||^2 
\ge\!\!\! \sum_{\substack{p:p=(x,\ldots,y) \\ x\in A-A',B\in B-B'}} \!\!\! f_p ||x-y||^2 
\ge  \tfrac 13 c'\pi n \cdot \Delta = 2\alpha
$$
But then the algorithm should have terminated at line 7 - a contradiction.
\end{proof}

In the remainder of the analysis we assume that case (i) holds in Lemma~\ref{lemma:matching}(b).
(In the case of case (ii) procedures that we will describe will terminate the oracle at line 5 or 7 with probability $\Theta(1)$).

We assume that procedure ${\tt Matching}(\cdot)$ satisfies the following skew-symmetry condition:
 for any $u$, matching ${\tt Matching}(-u)$ is obtained from ${\tt Matching}(u)$ by reversing all edge
orientations. (This can be easily enforced algorithmically).

\subsection{Matching covers}

We say that a {\em generalized matching} is a set $M$ of paths of the form $p=(p_0,\ldots,p_k)$ with $p_0,\ldots,p_k\in V$
such that each node $x\in V$ has at most one incoming and at most one outgoing path.
 We say that path $q$ is {\em violating}
if  $q=(\ldots,p,\ldots)$ and path $p=(p_0,\ldots,p_{\ell(p)})$ satisfies~\eqref{eq:PathViolationtilde}.
We denote $M^{\tt violating}$ and $M^{\tt nonviolating}$ to be the set of violating and nonviolating paths in
generalized matching $M$, respectively.
Below we will only be interested in the endpoints of paths $p\in M$ and violating/nonviolating status of $p$.
Thus, paths in $M$ will essentially be treated as edges with a Boolean flag.
For generalized matchings $M_1,M_2$ we define
$$
M_1\circ M_2=\{(p,x,q)\::\:(p,x)\in M_1,(x,q)\in M_2\}
$$
where $p,q$ are paths and $x$ is a node. 
Clearly, $M_1\circ M_2$ is also a generalized matching.
For generalized matching $M$, vector $u\in\mathbb R^d$ and value $\sigma\in\mathbb R$ we define
\begin{align*}
{\tt Truncate}_\sigma(M;u)&=M^{\tt violating}\;\;\cup\;\; \{(x,\ldots,y)\in M^{\tt nonviolating}\::\:\langle y-x, u\rangle\ge \sigma\}
\end{align*}

%Given a generalized matching $M$ and value $\sigma\in\mathbb R$, we define
%$$
%{\tt Truncate}_\sigma(M;u)=M^{\tt violating}\cup\{(x,\ldots,y)\in M^{\tt nonviolating}\::\:\langle y-x,u\rangle\ge \sigma\}
%$$

We will consider algorithms for constructing generalized matchings that have
 the following form: given vector $u\in\mathbb R^d$,
 sample vectors $(u_1,\ldots,u_k)$ according to some distribution that depends on $u$,
and return ${\tt Matching}(u_1)\circ\ldots\circ{\tt Matching}(u_k)$.
Any such algorithm specifies a {\em matching cover} as defined below.

\begin{definition}
A {\em generalized matching cover} (or just {\em matching cover})
is function $\calM$ that maps vector $u\in\mathbb R^d$ to a distribution over generalized matchings.
It is called {\em skew-symmetric} if
sampling $M\sim\calM(-u)$
and then reversing all paths in $M$ produces the same distribution as sampling $M\sim\calM(u)$.
%\begin{itemize}
%\item For every $u\in\mathbb R$ matching $M\sim\calM(u)$ has the same distribution
%as the matching obtained by sampling $M\sim\calM(-u)$ and reversing all edges in $M$.
%\end{itemize}
We define ${\tt size}(\calM)=\tfrac 1n\mathbb E_{M\sim \calM(\calN)}[|M|]$
%the expected size of $M\sim\calM(\calN^d)$
where $\calM(\calN)$ denotes the distribution $u\sim \calN,M\sim\calM(u)$.
%Matching cover $\calM$ is said to be {\em skew-symmetric} if sampling $M\sim\calM(-u)$
%and then reversing all paths in $M$ produces the same distribution as sampling $M\sim\calM(u)$.
We say that $\calM$ is {\em $\sigma$-stretched} (resp. {\em $L$-long}) 
if $\langle y-x,u\rangle\ge \sigma$
(resp. $||y-x||^2\le L$)
for any $u\in\mathbb R^d$ and any nonviolating $(x,\ldots,y)\in {\tt supp}(\calM(u))$.
$\calM$ is {\em $k$-hop} if any path $p\in M\in{\tt supp}(\calM(\calN))$ has the form $p=(p_0,p_1,\ldots,p_k)$
where $||p_i-p_{i-1}||^2\le \Delta$ for all $i\in[k]$.
%consists of exactly $k$ edges.
We write $\calM\subseteq\calM'$ for (coupled) matching covers $\calM,\calM'$
if $M\subseteq M'$ for any $u$ and $M\sim\calM(u)$, $M'\sim\calM(u')$.
\end{definition}
If $\calM$ is a matching cover then we define matching cover $\calM_{[\sigma]}$ as follows:
given vector $u$, sample $M\sim\calM(u)$ and let ${\tt Truncate}_\sigma(M;u)$ be the output
of $\calM_{[\sigma]}(u)$. Clearly, $\calM_{[\sigma]}$ is $\sigma$-stretched.

The following lemma shows how matching covers can be used. We say that direction $u\in\mathbb R^d$ is {\em regular}
if $\langle y-x,u\rangle< \sqrt{6\ln n}\cdot ||y-x||$ for all distinct $x,y\in V$.

\begin{lemma}\label{lemma:regular}
Let $\calM$ be a $k$-hop matching cover. \\
%{\rm (a)} Any nonviolating path $p=(x,\ldots,y)\in M\in{\tt supp}(\calM(\calN))$ satisfies $||y-x||^2\le(k+1)\Delta$. \\
{\rm (a)} $\calM$ is $((k+1)\Delta)$-long. \\
{\rm (b)} If $\calM$ is $\sqrt{6(k+1)\Delta\ln n}$-stretched
then $M^{\tt nonviolating}\!=\!\varnothing$ for any regular $u$ and  $M\!\in\!{\tt supp}(\calM(u))$. \\
{\rm (c)}  ${\tt Pr}_u[u\mbox{ is regular\,}]\ge 1-\tfrac 1n$.
\end{lemma}
\begin{proof}
\myparagraph{(a)} By definitions, any nonviolating path $p=(p_0,p_1,\ldots,p_k)\in M\in{\tt supp}(\calM(\calN))$
satisfies $||p_k-p_0||^2\le\Delta+\sum_{i=1}^k||p_i-p_{i-1}||^2\le \Delta+\sum_{i=1}^k\Delta=(k+1)\Delta$.

\myparagraph{(b)} Suppose there exists nonviolating path $(x,\ldots,y)\in M\in{\tt supp}(\calM(u))$.
Elements $x,y$ must be distinct.
Part (a) gives $||y-x||\le \sqrt{(k+1)\Delta}$.
Regularity of $u$ thus implies that $\langle y-x,u\rangle<\sqrt{6\ln n}\cdot \sqrt{(k+1)\Delta}$ - a contradiction.

\myparagraph{(c)} Consider distinct $x,y\in V$. 
The quantity $\langle y-x, u\rangle$ is normal with zero mean
and variance $||y-x||^2$ under $u\sim\calN$.
Thus, 
$
{\tt Pr}_u[\langle y-x, u\rangle \ge c||y-x||]
={\tt Pr}[\calN(0,||y-x||)\ge c||y-x||]
={\tt Pr}[\calN(0,1)\ge c]\le e^{-c^2/2}
=\tfrac 1 {n^3}
$ for $c=\sqrt{6\ln n}$. There are at most $n^2$ distinct pairs $x,y\in V$, so the union bound gives the claim.
\end{proof}

\subsection{Chaining algorithms}\label{sec:chaining}
By construction, ${\tt Matching}$ is a (deterministic) 1-hop $\sigma$-stretched skew-symmetric matching cover.
%with ${\tt size}({\tt Matching})\ge \delta$.
Next, we discuss how to ``chain together'' matchings returned by ${\tt Matching}(\cdot)$ to obtain
a matching cover with a large stretch, as required by Lemma~\ref{lemma:regular}(b).

\myparagraph{Sherman's algorithm}
First, we review Sherman's algorithm~\cite{Sherman}.
In addition to vector $u$, it takes a sequence  $b=(b_1,\ldots,b_K)\in\{0,1\}^K$ as an input.

\begin{algorithm}[H]
  \DontPrintSemicolon
\SetNoFillComment
	let $k$ be the number of $1$'s in $b$ and $k'$ be the number of $0$'s in $b$, so that $k+k'=K$ \\
	sample $(u_1,\ldots,u_k)\sim \calN_\omega^k|u_1$ and $(u'_1,\ldots,u'_{k'})\sim \calN_0^{k'}$ independently where $\omega=1-1/k$ \\
let $(\bar u_1,\ldots,\bar u_{K})$ be the sequence
obtained by merging sequences $(u_1,\ldots,u_k)$ and $(u'_1,\ldots,u'_{k'})$ at positions specified by $b$
in a natural way \\
return ${\tt Matching}(\bar u_1)\circ\ldots\circ {\tt Matching}(\bar u_K)$ 
      \caption{${\tt SamplePaths}^b(u_1)$. 
      }\label{alg:Sherman}
\end{algorithm}

\begin{theorem}\label{th:Sherman}
\begin{sloppypar}
For any $\delta>0$ and $\sigma>0$ there exist positive constants $c_1,c_2,c_3$ with the following property.
Suppose that ${\tt Matching}$ is a 1-hop $\sigma$-stretched skew-symmetric matching cover with ${\tt size}({\tt Matching})\ge \delta$, and $K\Delta\le c_1$.
Then there exists vector $b\in\{0,1\}^K$ for which ${\tt size}\left({\tt SamplePaths}^b_{[c_2 K]}\right)\ge e^{-c_3K^2}$.
\end{sloppypar}
\end{theorem}
In practice vector $b$ is not known, however it can be sampled uniformly at random
which decreases the expectation by a factor of $2^{-K}$, which does not change the bound $e^{-\Theta(K^2)}$ in Theorem~\ref{th:Sherman}.
Note that~\cite{Sherman} does not use the notion of ``violating paths'', and accordingly Theorem~\ref{th:Sherman}
is formulated slightly differently in~\cite{Sherman}. However, the proof in~\cite{Sherman} can be easily adapted to yield 
Theorem~\ref{th:Sherman}.~\footnote{
We believe that \cite{Sherman} has a numerical
mistake. Namely, consider value $\delta\in(0,1)$ and
weighted graph $(V,E,w)$ with $n=|V|$ nodes and non-negative weights $w$ 
such that $w(x,A)=w(A,x)\le 1$ for all $x\in V$ and $A\subseteq V$ 
(where $w(X,Y)=\sum_{(x,y)\in X\times Y}w(x,y)$). \cite[proof of Lemma 4.8]{Sherman} essentially makes
the following claim:
\begin{itemize}[leftmargin=15pt]
\item {\em
For any subset $B\subseteq V$ with $w(V,B)\ge \delta|B|$
 there exists $A\subseteq V$
such that either (i)~$|A|\ge  \delta |B|$ and $w(x,B)\ge \delta^3$ for any $x\in A$,
or (ii) $|A|\ge \tfrac 1{\delta} |B|$ and $w(x,B)\ge \tfrac {\delta^2}{n}|B|$ for any $x\in A$.
}
\end{itemize}
A counterexample can be constructed as follows.
Assume that $\delta|B|$ and $\tfrac{1}{\delta^4}$ are integers. 
For $\delta |B|-1$ nodes $x$ set $w(x,B)=1$, for $\tfrac 1{\delta^4}$ nodes $x$ set $w(x,B)=\delta^4$,
and for remaining nodes set $w(x,B)=0$.
Then the claim above is false if $\delta |B|-1+\tfrac 1 {\delta^4}<\tfrac 1\delta|B|$.
The statement can be corrected as follows:

\begin{itemize}[leftmargin=15pt]
\item {\em 
For any $\lambda\le\delta$ and any subset $B\subseteq V$ with $w(V,B)\ge \delta|B|$
there exists $A\subseteq V$
such that either (i) $|A|\ge \tfrac  {\delta}3 |B|$ and $w(x,B)\ge \lambda$ for any $x\in A$,
or (ii) $|A|\ge \tfrac {\delta} {3\lambda}|B|$ and $w(x,B)\ge \tfrac {\delta}{3n}|B|$ for any $x\in A$.
(If both conditions are false then $w(V,B)< \tfrac {\delta}3|B|\cdot 1+\tfrac {\delta} {3\lambda}|B|\cdot\lambda+n\cdot \tfrac {\delta}{3n}|B|=\delta|B|$,
which is impossible since $w(V,y)\ge  \delta$ for each $y\in B$).
}
\end{itemize}
Then the proof in~\cite{Sherman} still works, but with different constants.
}

\myparagraph{Robustness to deletions and parallelization} Note that Theorem~\ref{th:Sherman} 
can also be applied to any matching cover ${\tt Matching}'\subseteq{\tt Matching}$ satisfying ${\tt size}({\tt Matching}')\ge\delta'$
for some positive constant $\delta'<\delta$. Thus, we can adversarially ``knock out'' some edges from ${\tt Matching}$
and still get useful bounds on the size of the output.
This is a key proof technique in this paper; to our knowledge, it has not been exploited before.
Our first use of this technique is for parallelization.
We need to show that running Algorithm~\ref{alg:Sherman} multiple times independently produces
many disjoint paths.
Roughly speaking, our argument is as follows.
Consider the $i$-th run, and assume that $|M_{\tt prev}|$ is small
where $M_{\tt prev}$  is the union of paths computed in the first $i-1$ runs.
Let us apply Theorem~\ref{th:Sherman} to the matching
cover obtained from  ${\tt Matching}$ by knocking out edges incident to nodes in $M_{\tt prev}$.
It yields that the $i$-th run produces many paths that are node-disjoint from $M_{\tt prev}$, as desired.
We refer to Section~\ref{sec:FinalAlg} for further details.

\myparagraph{New chaining algorithm} 
Next, we discuss how
a similar proof technique (combined with additional ideas) can be used
to simplify Sherman's chaining algorithm.
We will show that Theorem~\ref{th:Sherman}  still holds if we consider vectors $b$ of the form $b=(1,\ldots,1)$.
The algorithm thus becomes as follows.

\begin{algorithm}[H]
  \DontPrintSemicolon
\SetNoFillComment
	sample $(u_1,\ldots,u_K)\sim \calN_\omega^K|u_1$  where $\omega=1-1/K$ \hspace{10pt} 
	\tcp*{\small equivalently, sample $u_2\!\sim_\omega \!u_1$, $u_3\!\sim_\omega\! u_2, ~\ldots,~ u_K\!\sim_\omega \!u_{K-1}$\!\!\!\!\!\!\!\!\!\!\!} 
	return ${\tt Matching}( u_1)\circ\ldots\circ {\tt Matching}( u_K)$ 
      \caption{${\tt SamplePaths}^K(u_1)$. 
      }\label{alg:vnk}
\end{algorithm}

\begin{theorem}\label{th:vnk}
\begin{sloppypar}
For any $\delta>0$ and $\sigma>0$ there exist positive constants $c_1,c_2,c_3$ with the following property.
Suppose that ${\tt Matching}$ is a 1-hop $\sigma$-stretched skew-symmetric matching cover with ${\tt size}({\tt Matching})\ge \delta$, $K\Delta\le c_1$,
and $K=2^r$ for some integer $r$.
Then ${\tt size}\left({\tt SamplePaths}^K_{[c_2 K]}\right)\ge e^{-c_3K^2}$.
\end{sloppypar}
\end{theorem}

We prove this theorem in Section~\ref{sec:proof:vnk}. Our proof technique is very different from~\cite{Sherman},
and relies on two key ideas: 
\begin{enumerate}
\item We use induction on $k=2^0,2^1,2^2,\ldots,K$ to show that ${\tt SamplePaths}^k_{[\sigma_k]}$ has a sufficiently
large size assuming that ${\tt size}({\tt Matching})\ge \delta_k$, for some sequences $\delta_1<\delta_2<\delta_4<\ldots<\delta_K=\delta$
and $\sigma_1<\sigma_2<\sigma_4<\ldots<\sigma_K=\Theta(K)$.
To prove the claim for $2k$, we show that for each node $x$, function $\mu_x(u)$ is ``sufficiently spread'' 
where $\mu_x(u)$ is the expected out-degree of $x$ in $M\sim {\tt SamplePaths}^k_{[\sigma_k]}(u)$.
In order to do this, we ``knock out'' edges $(x,y)$ from ${\tt Matching}(u)$ for a $\Theta(\delta_{2k}-\delta_k)$ fraction of vectors $u$ with the largest value of $\mu_x(u)$,
and then use the induction hypothesis for the smaller matching cover of size $\delta_k$. 
We conclude that $\mu_x(u)$ is sufficiently large for $\Theta(\delta_{2k}-\delta_k)$ fraction of $u$'s.
We then use Theorem~\ref{th:isoperimetric} and skew-symmetry to argue that chaining ${\tt SamplePaths}^k_{[\sigma_k]}$ with itself
gives a matching cover of large size.
\item We work with ``extended matching covers'' instead of matching covers. These are functions $\calM$ 
that take a {\bf pair} of vectors $(u_1,u_k)\in\mathbb R^d\times\mathbb R^d$ as input,
sample $(u_1,\ldots,u_k)\sim\calN^k_\omega$ conditioned on fixed $u_1,u_k$,
and return (a subset of) ${\tt Matching}(u_1)\circ\ldots\circ {\tt Matching}(u_k)$.
This guarantees that if $(x,y)$ is removed from ${\tt Matching}(u_1)$ then $x$ has no outgoing edge in $\calM(u_1,u_k)$
(which is needed by argument above).
\end{enumerate}

Our proof appears to be more compact than Sherman's proof, and also does not rely on case
analysis. Accordingly, we believe that our technique should give smaller constants in the $O(\cdot)$ notation
 (although neither proof explicitly optimizes these constants).

Next, we discuss implications of Theorem~\ref{th:vnk}.
Below we denote $\llceil  z\rrceil=2^{\lceil\log_2 z\rceil}$ to be the smallest $K=2^r$, $r\in\mathbb Z_{\ge 0}$ satisfying $K\ge z$.
\begin{corollary}\label{cor:vnk}
Suppose that ${\tt Matching}$ is a 1-hop $\sigma$-stretched skew-symmetric matching cover with ${\tt size}({\tt Matching})\ge \delta$.
Define  $A=\tfrac{12}{c_2^2}$, $B=\frac 1{2A\sqrt{c_3}}$ and $\varepsilon^{\max}=\frac{c_1}{2AB^2}$.
Suppose that $\Delta= B\sqrt{\frac{\varepsilon}{\ln n}}$
where $\varepsilon\in(0,\varepsilon^{\max}]$ and $A\Delta\ln n\ge 1$.
%~\footnote{Condition $A\Delta\ln n\ge 1$ means that $\varepsilon\ln n=\Omega(1)$; if $\varepsilon$ is fixed then it will hold for sufficiently large $n$.}
If $K=\llceil A \Delta\ln n\rrceil$ then 
$$
\mathbb E_{M\sim{\tt SamplePaths}^K(\calN)} [| M^{\tt violating} |] \ge 
n^{1-\varepsilon}-1
$$
\end{corollary}
\begin{proof}
Note that $K\in[A\Delta\ln n,2 A\Delta\ln n]$,
and
$
K\Delta \le   2A\Delta^2\ln n
=2AB^2\varepsilon
\le 2AB^2\varepsilon^{\max}
=c_1
$.
Denote $\calM={\tt SamplePaths}^K_{[c_2 K]}$.
Theorem~\ref{th:vnk} gives
$
{\tt size}(\calM)\ge 
e^{-c_3K^2}
\ge e^{-c_3(2A\Delta \ln n)^2}
=n^{-\varepsilon}
$.
We also have
$
\frac{6 (K+1)\Delta\ln n }{(c_2K)^2}
\le \frac {12 \Delta\ln n}{c_2^2 K}
\le \frac{12 \Delta \ln n}{c_2^2 A\Delta\ln n}
=1
$
and so the precondition of Lemma~\ref{lemma:regular}(b) holds for $\calM$.

\begin{sloppypar}
Let
$x=\mathbb E_{u\sim\calN|u\mbox{\scriptsize~is regular},M\sim\calM(u)}[|M|]
=\mathbb E_{u\sim\calN|u\mbox{\scriptsize~is regular},M\sim\calM(u)}[|M^{\tt violating}|]
$
(by Lemma~\ref{lemma:regular}(b)),
 $y=\mathbb E_{u\sim\calN|u\mbox{\scriptsize~is not regular},M\sim\calM(u)}[|M|]\le n$
 and $p={\tt Pr}_u[u\mbox{ is not regular}]\le \tfrac 1n$
 (by Lemma~\ref{lemma:regular}(c)).
 We have
 $
 n^{1-\varepsilon}\le \mathbb E_{M\sim\calM(\calN)}[|M|] = (1-p)x+py
 \le (1-p)x+\tfrac 1n\cdot n
 $ and so $(1-p)x\ge n^{1-\varepsilon}-1$.
Therefore,
$
\mathbb E_{M\sim{\tt SamplePaths}^K(\calN)} [ | M^{\tt violating} | ] \ge 
(1-p)
x\ge
n^{1-\varepsilon}-1
$.
\end{sloppypar}
\end{proof}

%%%%%%%%%%%%%%%%%%%%%%%%%%%%%%%%%%%%%%%%%%%%%%%%%%%%%%%%%%%%%%%%%%%%%%%%%%%%%%%%%%%%%%%%%%%%%%%%%%%
%%%%%%%%%%%%%%%%%%%%%%%%%%%%%%%%%%%%%%%%%%%%%%%%%%%%%%%%%%%%%%%%%%%%%%%%%%%%%%%%%%%%%%%%%%%%%%%%%%%
%%%%%%%%%%%%%%%%%%%%%%%%%%%%%%%%%%%%%%%%%%%%%%%%%%%%%%%%%%%%%%%%%%%%%%%%%%%%%%%%%%%%%%%%%%%%%%%%%%%
%%%%%%%%%%%%%%%%%%%%%%%%%%%%%%%%%%%%%%%%%%%%%%%%%%%%%%%%%%%%%%%%%%%%%%%%%%%%%%%%%%%%%%%%%%%%%%%%%%%
%%%%%%%%%%%%%%%%%%%%%%%%%%%%%%%%%%%%%%%%%%%%%%%%%%%%%%%%%%%%%%%%%%%%%%%%%%%%%%%%%%%%%%%%%%%%%%%%%%%

\subsection{Final algorithm}\label{sec:FinalAlg}

In the algorithm below we use the following notation:
$V(p)\subseteq V$ is the set of nodes through which path $p$ passes,
and $V(M)=\bigcup_{p\in M} V(p)$.
Furthermore, if $p$ is violating then $p^{\tt violating}$ is a subpath of $p$
satisfying~\eqref{eq:PathViolation}. Note that lines 1-3 compute sets of paths $\tilde M_1,\ldots,\tilde M_N$,
which are then combined into a single set $M\subseteq\bigcup_i \tilde M_i$ using one of the two options.
Option 1 will be mainly used for the analysis, while option 2 will be used for an efficient parallel implementation.

\ignore{
\begin{algorithm}[H]
  \DontPrintSemicolon
\SetNoFillComment
set $M:=\varnothing$ \\
\For{$i=1,\ldots,N$}
{
	sample $u\sim\calN$, call $M_i={\tt SamplePaths}^K(u)$ \\
	let $\tilde M_i=\{p^{\tt violating}\::\:p\in M_i^{\tt violating}\}$ \\
	let $M := M \cup \{p\in \tilde M_i \::\: V(p)\cap V(M)=\varnothing\}$
}
	return $M$ 
      \caption{Computing violating paths. 
      }\label{alg:parallel}
\end{algorithm}
}

\begin{algorithm}[H]
  \DontPrintSemicolon
\SetNoFillComment
\For{$i=1,\ldots,N$}
{
	sample $u\sim\calN$, call $M_i={\tt SamplePaths}^K(u)$ \\
	let $\tilde M_i=\{p^{\tt violating}\::\:p\in M_i^{\tt violating}\}$ 
}
option 1: set $M\!:=\!\varnothing$, then for $i\!=\!1,\ldots,N$ update $M \!:=\! M \cup \{p\in \tilde M_i \::\: V(p)\cap V(M)=\varnothing\}$ \\
option 2: let $M$ be a maximal set of paths in $M^\ast\eqdef\bigcup_{i=1}^N\tilde M_i$ s.t.\ $V(p)\!\cap\! V(q)=\varnothing$ for $p,q\!\in\! M$, $p\!\ne\! q$ \\
	return $M$ 
      \caption{Computing violating paths. 
      }\label{alg:parallel}
\end{algorithm}

%In the theorem below we use constants $A,B_\delta,\varepsilon^{\max}_\delta$ defined in Corollary~\ref{cor:vnk}.
\begin{theorem}\label{th:parallel}
Suppose that ${\tt Matching}$ is a 1-hop $\sigma$-stretched skew-symmetric matching cover with ${\tt size}({\tt Matching})\ge \delta$,
and let $A,B,\varepsilon^{\max}$ be the constants defined on Corollary~\ref{cor:vnk} for value $\delta/2$.
Suppose that $\Delta= B\sqrt{\frac{\varepsilon}{\ln n}}$
where $\varepsilon\in(0,\varepsilon^{\max}]$ and $A\Delta\ln n\ge 1$.
If $K=\llceil A \Delta\ln n\rrceil$ and $N\ge \tfrac{ \delta n^\varepsilon}{4 K(1-n^{\varepsilon-1})} $ then $\mathbb E[|M|]\ge \tfrac{\delta n}{8K}$  where $M$ is the output of Algorithm~\ref{alg:parallel} with option 1.
\end{theorem}

\begin{proof}
Let $\calE_i$ be the event that set $M$ at the beginning of iteration $i$ satisfies $|M|\le a:=\tfrac{ \delta n}{4 K}$,
and let $\gamma_i={\tt Pr}[\calE_i]$.
Let $x_i$ be the expected number of paths that have been added to $M$ at iteration $i$, so that $\mathbb E[|M|]=x_1+\ldots+x_N$ for the final set $M$.
Let $y_i$ be the expected number of paths that have been added to $M$ at iteration $i$ conditioned on event $\calE_i$.
Clearly, we have $x_i \ge \gamma_i y_i \ge \gamma y_i$ where $\gamma:=\gamma_N$.

Next, we bound $y_i$. Let $M$ be the set at the beginning of iteration $i$, and suppose that $\calE_i$
holds, i.e.\ $|M|\le \tfrac{ \delta n}{4 K}$.
Denote $U=V(M)$, then $|U|\le K|M|\le\tfrac{ \delta n}{4}$.
Let ${\tt Matching}'\subseteq{\tt Matching}$ be the (skew-symmetric) matching cover obtained by removing from ${\tt Matching}(u)$
edges $(x,y)$ and $(y,x)$ with $x\in U$ (for all $u\in\mathbb R^d$).
Clearly, we have ${\tt size}({\tt Matching}')\ge\delta-\tfrac 12\delta=\tfrac 12\delta$.
Let ${\tt SamplePaths}'^K$ be the matching cover given by Algorithm~\ref{alg:vnk}
where ${\tt Matching}$ is replaced by ${\tt Matching}'$. Clearly,
we have ${\tt SamplePaths}'^K\subseteq {\tt SamplePaths}^K$.
By Corollary~\ref{cor:vnk} applied to ${\tt Matching}'$,
 ${\tt SamplePaths}'^K(\calN)$ produces at least $n^{1-\varepsilon}-1$ violating paths in expectation.
By construction, all these paths $p$ satisfy $V(M)\cap V(p)=\varnothing$,
and therefore $y_i\ge n^{1-\varepsilon}-1$.

We showed that $\mathbb E[|M|]\ge \sum_{i=1}^N \gamma(n^{1-\varepsilon}-1)=\gamma b$  for the final set $M$ where $b:=N(n^{1-\varepsilon}-1)$.
We also have $\mathbb E[|M|]\ge (1-\gamma)a$,
and hence $\mathbb E[|M|]\ge \min_{\gamma\in[0,1]} \max\{\gamma b,(1-\gamma)a\}=\tfrac{ab}{a+b}$
(the minimum is attained at $\gamma=a/(a+b)$). By assumption, we have $b\ge a$, and so 
$\mathbb E[|M|] \ge \tfrac 12 a$.
\end{proof}

Next, we analyze Algorithm~\ref{alg:parallel} with option 2. Recall that ``parallel runtime'' is the maximum runtime over all processors.
\begin{lemma}\label{lemma:option2}
(a) Let $M$ and $M'$ be the outputs of Algorithm~\ref{alg:parallel} with options 1 and 2, respectively
(for a given run of the loop in lines 1-3). Then $|M'|\ge |M|/K^3$. \\
(b) Algorithm~\ref{alg:parallel} with option 2 can be implemented on $N$ processors with parallel runtime $O(KT_{\tt maxflow}+n(Kd+K\log^2 n + \log N)))$
(in any version of PRAM).
\end{lemma}
\begin{proof} 
\myparagraph{(a)} 
By the construction of $M$, each node $v\in V(M^\ast)$ is contained in at most $K$ paths $p\in M$
(all of them belong to some $\tilde M_i$ for fixed $i$).
Therefore, $|V(M^\ast)|\ge |M|/K$.
We also have $|V(M')|\ge |V(M^\ast)|/ K$ and $|M'|\ge|V(M')|/ K$, since $|V(p)|\le K$ for any $p\in M^\ast$.
Putting these inequalities together gives the claim.

\myparagraph{(b)} Clearly, sets $\tilde M_i$ for $i\in[N]$ can be computed in parallel on $N$ processors
in time $O(K(T_{\tt maxflow}+nd+n\log^2 n)))$ per processor
 (since computing dot products $\langle x,u\rangle$ in Algorithm~\ref{alg:Matching} takes time $O(Knd)$,
 and flow decompositions take time $O(Km\log n)=O(Kn\log^2 n)$).
 We can assume that after computing $\tilde M_i$, processor $i$
 computes a maximal set of paths $\tilde M_i'\subseteq\tilde M_i$ that are pairwise node-disjoint, and updates $\tilde M_i:=\tilde M_i'$.
 Clearly, this step does not affect the output of line 5.
 
 It remains to discuss how to implement line 5 (computing a maximal set of paths $M$ in $M^\ast$ which are pairwise node-disjoint).
 For $N=2$ this can be easily done in $O(n)$ time: processor 1 sends $\tilde M_1$ to processor 2, and processor 2 computes the answer.
 The general case can be reduced to the case above using a divide-and-conquere strategy with a computation tree which is a binary
 tree whose leaves are the $N$ processors. The depth of this tree is $O(\log N)$, and hence the parallel runtime of this procedure is $O(n\log N)$.
\end{proof}

Note in Algorithm~\ref{alg:parallel} 
step 3 can be run in parallel on $N$ processors; after all of them finish, we can run the rest of algorithm on a single machine.
By putting everything together, we obtain
\begin{theorem}
There exists an algorithm for {\sc Balanced Separator}
that given $\varepsilon\in[\Theta(1/\log n),\Theta(1)]$, produces $O(\sqrt{(\log n)/\varepsilon})$-pseudoapproximation w.h.p..
Its expected parallel runtime is $O((\log^{O(1)}n) T_{\tt maxflow} )$ on $O(n^\varepsilon)$ processors (in any version of PRAM).
\end{theorem}
\begin{proof}

Let $\delta,\sigma$ be as in Lemma~\ref{lemma:matching}.
Set parameters as in Theorem~\ref{th:parallel}, and require additionally that $\varepsilon\le \tfrac 12$.
Condition $A\Delta\ln n\ge 1$ means that this can be done for $\varepsilon\in[\Theta(1/\log n),\Theta(1)]$.
By Theorem~\ref{th:parallel} and Lemma~\ref{lemma:option2}, the output of Algorithm~\ref{alg:parallel}
satisfies $\mathbb E[|M|]\ge \frac{\delta n}{8 K^{1+h}}$
where $h=0$ if option 1 is used, and $h=3$ if option 2 is used.

To implement the oracle, run Algorithm~\ref{alg:parallel}
until either procedure ${\tt Matching}(\cdot)$ (Alg.~\ref{alg:Matching}) terminates at lines 5 or 7,
or until we find a set $M$ of violating paths with $|M|\ge \tfrac {\delta n}{16 K^{1+h}}$.
In the latter case use Lemma~\ref{lemma:ViolatingPaths} to set the variables.
Since we always have $|M|\le n$, the expected number of runs will be $O(n/\tfrac n{K^{1+h}})=O(K^{1+h})$.
%Clearly, line 3 of Algorithm~\ref{alg:parallel} for $i=1,\ldots,N$
%can be implemented in parallel, and so the oracle's complexity is $O(K)$ maxflows on $N=O(n^\varepsilon)$ processors.

If the oracle terminates at line 5 of Alg.~\ref{alg:Matching}, then it returns a $c'$-balanced cut of cost at most $\tfrac{6\alpha}\Delta=O\left(\alpha\sqrt{\tfrac{\log n}\varepsilon}\right)$.
Otherwise it returns valid variables $f_p,F$. 
If the oracle terminates at line 7 then its width is $\rho=O(\tfrac \alpha n + \pi)=O(\tfrac {\alpha}{ n \Delta})$.
Now suppose that the oracle finds set $M$ of violating paths with $|M|\ge \tfrac {\delta n}{16 K^{1+h}}$.
Clearly, degrees $\pi_F,\pi_D$ in Lemma~\ref{lemma:ViolatingPaths} satisfy $\pi_F=O(K)$ and $\pi_D=O(1)$,
so the oracle's width in this case is $\rho=O(\tfrac \alpha n + \tfrac {\alpha K^{2+h}}{n\Delta})$.
In both cases we have $\rho=O(\tfrac {\alpha K^{2+h}}{n\Delta})$.
Thus, the number of calls to ${\tt Oracle}$
 %iterations of Algorithm~\ref{alg:MW} 
 for a fixed value of $\alpha$ is $T=O(\tfrac{\rho^2n^2\log n}{\alpha^2})=O(\tfrac{K^{4+2h}\log n}{\Delta^2})$.

Next, we bound the complexity of computing approximations $\tilde v_1,\ldots,\tilde v_{n_{\tt orig}}$ for fixed $\alpha$ as described in Theorem~\ref{th:GramApproximation}
in Section~\ref{sec:Gram}. (Here we assume familiarity with Appendix~\ref{sec:app:MW}).
We have $\tau=\Theta(1)$ and $\gamma=\Theta\{\min\{\tfrac{\alpha}{n\pi},\tfrac{\Delta}{K}\})=\Theta(\tfrac\Delta K)$,
thus we need to use dimension $d=\Theta(\tfrac{\log n}{\gamma^2})=\Theta(\tfrac{K^2\log n}{\Delta^2})$.
We need to compute $Tkd$ matrix-vector products of form $A\cdot u$ where
$k=O(\max\{(\tfrac{\rho n\log n}{\alpha})^2,\log n\})=O(\tfrac{K^4\log^2 n}{\Delta^2})$.
Each matrix $A$ has the form $\sum_{r=1}^{O(T)} N^{(r)}$,
and each $N^{(r)}$ can be represented as a sum of $O(K)$ ``easy'' matrices (e.g.\ corresponding to matchings)
for which the multiplication with a vector takes $O(n)$ time.
To summarize, the overall complexity of 
computing approximations $\tilde v_1,\ldots,\tilde v_{n_{\tt orig}}$ is $Tkd\cdot TKn=O(\tfrac{K^{15+4h}\log^5 n}{\Delta^8}n)$, %=O(n\log^{O(1)}n)$.
which is $O(n\log^{O(1)}n)$
since $\Delta=\Theta\left(\sqrt{\tfrac{\varepsilon}{\log n}}\right)\in \left[\Theta(\tfrac{1}{\log n}),\Theta(\tfrac{1}{\sqrt{\log n}})\right]$ 
and $K=\Theta(\Delta\log n)=O(\sqrt{\log n})$.
These computations can be done on a single processor; their runtime is subsumed by the claimed $O((\log^{O(1)}n) T_{\tt maxflow} )$ bound.

We saw that $d=O(\log^{O(1)}n)$.
From Lemma~\ref{lemma:option2} we can now conclude that the algorithm's expected parallel runtime is
$T\cdot O(KT_{\tt maxflow}+n(Kd+K\log^2 n + \log N)))=O((\log^{O(1)}n) T_{\tt maxflow} )$
on $N=\Theta(n^\varepsilon/K)$ processors, assuming that option 2 is used.
Note that the output of the overall algorithm is correct w.h.p.\ since vectors 
$\tilde v_1,\ldots,\tilde v_{n_{\tt orig}}$ approximate original vectors only w.h.p.\ (see Theorem~\ref{th:GramApproximation}).
\end{proof}

%%%%%%%%%%%%%%%%%%
%%%%%%%%%%%%%%%%%%%%%%%%%%%%%%%%%%%%%%%%%%%%%%%%%%%%%%%%%%%%%%%%%%%%%%%%%%%%%%%%%%%%%%%%%%%%%%%%%%%
%%%%%%%%%%%%%%%%%%%%%%%%%%%%%%%%%%%%%%%%%%%%%%%%%%%%%%%%%%%%%%%%%%%%%%%%%%%%%%%%%%%%%%%%%%%%%%%%%%%
%%%%%%%%%%%%%%%%%%%%%%%%%%%%%%%%%%%%%%%%%%%%%%%%%%%%%%%%%%%%%%%%%%%%%%%%%%%%%%%%%%%%%%%%%%%%%%%%%%%
%%%%%%%%%%%%%%%%%%%%%%%%%%%%%%%%%%%%%%%%%%%%%%%%%%%%%%%%%%%%%%%%%%%%%%%%%%%%%%%%%%%%%%%%%%%%%%%%%%%

\section{Proof of Theorem~\ref{th:vnk}}\label{sec:proof:vnk}

We will need the following definition.
\begin{definition}
A {\em $k$-hop extended matching cover}
is function $\calM$ that maps vectors $u,u'\in\mathbb R^d$ to a distribution over generalized matchings.
It is {\em skew-symmetric}
if sampling $M\sim\calM(-u',-u)$
and then reversing all paths in $M$ produces the same distribution as sampling $M\sim\calM(u,u')$.
%\begin{itemize}
%\item For every $u\in\mathbb R$ matching $M\sim\calM(u)$ has the same distribution
%as the matching obtained by sampling $M\sim\calM(-u)$ and reversing all edges in $M$.
%\end{itemize}
We define ${\tt size}_\omega(\calM)=\tfrac 1n\mathbb E_{M\sim \calM(\calN_\omega)}[|M|]$
%the expected size of $M\sim\calM(\calN^d)$
where $\calM(\calN_\omega)$ denotes the distribution $(u,u')\sim \calN_\omega,M\sim\calM(u,u')$.
%Matching cover $\calM$ is said to be {\em skew-symmetric} if sampling $M\sim\calM(-u)$
%and then reversing all paths in $M$ produces the same distribution as sampling $M\sim\calM(u)$.
We say that $\calM$ is {\em $\sigma$-stretched} (resp. {\em $L$-long}) 
if $\min\{\langle y-x,u\rangle,\langle y-x,u'\rangle\}\ge \sigma$
(resp. $||y-x||^2\le L$)
for any $u,u'\in\mathbb R^d$ and any nonviolating $(x,\ldots,y)\in {\tt supp}(\calM(u,u'))$.
We write $\calM\subseteq\tilde\calM$ for (coupled) extended matching covers $\calM,\tilde \calM$
if $M\subseteq \tilde M$ for any $u,u'$ and $M\sim\calM(u,u')$, $\tilde M\sim\tilde\calM(u,u')$.
\end{definition}

For an extended matching cover $\calM$ we can define matching cover $\calM_\omega$ as follows:
given vector $u$, sample $u'\sim_\omega u$, $M\sim\calM(u,u')$ and let $M$ be the output of $\calM(u)$.
Clearly, if $\calM$ is $\sigma$-stretched then $\calM_\omega$ is also $\sigma$-stretched,
and ${\tt size}_\omega(\calM)={\tt size}(\calM_\omega)$.

We now proceed with the proof of Theorem~\ref{th:vnk}.
Define $\calK=\{2^0,2^1,2^2,\ldots,K\}$, and let $\dot \calK=\calK-\{K\}$.
(Recall that $K$ has the form $K=2^r$ for some integer $r$).
Let us choose positive numbers $\beta_k$ for $k\in\dot\calK$ (to be specified later),
and define numbers $\{\sigma_k\}_{k\in\calK}$ via the following
recursion:
\begin{subequations}\label{eq:MFSDA}
\begin{eqnarray}
\sigma_1&=&\sigma \\
\sigma_{2k}&=&(1+\omega^k)\sigma_k - \beta_k \qquad\forall k\in\dot\calK \label{eq:HLAJGA}
\end{eqnarray}
\end{subequations}

For a matching cover $\calM\subseteq{\tt Matching}$
and integer $k\in\calK$ we define extended matching cover $\calM^k$
as follows:
\begin{itemize}
\item given vectors $(u_1,u_k)$, sample $(u_1,u_2,\ldots,u_k)\sim\calN_\omega^k|(u_1,u_k)$,
compute $M=\calM(u_1)\circ\ldots\circ\calM(u_k)$ and let ${\tt Truncate}_{\sigma_k}(M;u_1)\cap {\tt Truncate}_{\sigma_k}(M;u_k)$
be the output of $\calM^k(u_1,u_k)$.
\end{itemize}
By definition, $\calM^k$ is $\sigma_k$-stretched. It can be seen that $(\calM^K)_{\omega^{K-1}}\subseteq {\tt SamplePaths}^K$
for any $\calM\subseteq{\tt Matching}$.
 Our goal will be to analyze
${\tt size}_{\omega^{K-1}}(\calM^K)={\tt size}((\calM^K)_{\omega^{k-1}})$ for $\calM={\tt Matching}$.

\begin{theorem}\label{th:induction}
Choose an increasing sequence of numbers $\{\delta_k\}_{k\in \calK}$ in $(0,1)$,
and define sequence $\{\lambda_k\}_{k\in\calK}$ via the following recursions:
\begin{subequations}\label{eq:GNAKGA}
\begin{eqnarray}
\lambda_1&\!\!\!\!\!=\!\!\!\!\!&\delta_1 \\
\lambda_{2k}&\!\!\!\!\!=\!\!\!\!\!&\theta_{2k}\lambda_k^2,\quad \theta_{2k}=\tfrac 12 \left(\tfrac 12(\delta_{2k}-\delta_k)\right)^{2/(1-\omega)} \quad\forall k\in\dot\calK \label{eq:MKASASF}
\end{eqnarray}
\end{subequations}
Suppose that
\begin{equation}
\exp\left(-\frac{\beta_k^2}{2(k+1)(1-\omega^k)\Delta}\right) \le \tfrac 12 \lambda_{2k} \qquad\forall k\in\dot\calK \label{eq:HDHGDHDGAHKA}
\end{equation}
Then for any $k\in\calK$ and any skew-symmetric matching cover $\calM\subseteq{\tt Matching}$ with ${\tt size}(\calM)\ge \delta_k$
there holds ${\tt size}_{\omega^{k-1}}(\calM^k)\ge\lambda_k$.

\end{theorem}

By plugging appropriate sequences $\{\delta_k\}$ and $\{\beta_k\}$
we can derive Theorem~\ref{th:vnk} as follows.
%(The proof of Lemma~\ref{lemma:GAGFHASF} is given in Appendix~\ref{sec:GAGFHASF}).

\begin{lemma}\label{lemma:GAGFHASF}
Define $\delta_k=(1-\tfrac 1{2k}) \delta$
and 
$
\beta_k = 2\phi\sigma k\sqrt{1-\omega^k}
$
where
$
 \phi = \tfrac {1}{\sigma} \left( K\Delta \cdot 5 \ln \tfrac {16}{ \delta} \right)^{1/2} 
$. Then~\eqref{eq:HDHGDHDGAHKA} holds, and
\begin{eqnarray}
\lambda_k & \ge & \left(\frac {16}{ \delta}\right)^{-2Kk} 
\end{eqnarray}
Furthermore, if $\phi\le \tfrac 18$ (or equivalently $K\Delta<\frac{\sigma^2}{320\ln \tfrac {16}{ \delta}}$ ) then
\begin{eqnarray}
\sigma_K & \ge & \tfrac 1 {16}\, \sigma \, K
\end{eqnarray}

\end{lemma}
\begin{proof}
Recall that $\omega=1-1/K$.
Let us set $\delta_k=(1-\tfrac 1{2k}) \delta$, then $\delta_1=\tfrac 12  \delta$
and $\theta_j=\tfrac 12 \left(\tfrac {\delta}{4j}\right)^{2K}=\tfrac {C}{j^{2K}}$ for $j\ge 2$
where $C=\tfrac 12 \left(\tfrac {\delta}{4}\right)^{2K}$.
%Plugging this into~\eqref{eq:GNALGHASGA} gives
By expanding expressions~(12) we obtain
%\begin{eqnarray}\label{eq:GNALGHASGA}
%\lambda_k&=&\delta_1^k\prod_{i\in\calK\cap[k/2]} \theta^i_{k/i}\quad\qquad\forall k\in\calK
%\end{eqnarray}
\begin{eqnarray*}
\lambda_k&=& \delta_1^k\prod_{i\in\calK\cap[k/2]} \theta^i_{k/i} \;\;=\;\;
\left(\frac{\delta}2\right)^k
\cdot\left( \prod_{i\in \calK\cap[k/2]}C^{i} \right)
\cdot
\left( \prod_{i\in \calK\cap[k/2]}(i/k)^{i} \right)^{2K} 
\\ &\stackrel{\mbox{\tiny (1)}}{>}& 
\left(\frac{\delta}2\right)^k
\cdot C^{k-1}
\cdot
\left( 2^{-2k} \right)^{2K} 
\;\;\stackrel{\mbox{\tiny (2)}}{>}\;\; (2C)^k\cdot 2^{-4Kk}
\;\;=\;\;\left(\frac {16}{ \delta}\right)^{-2Kk}
\qquad\quad\forall k\in\calK-\{1\}
\end{eqnarray*}
where in (2) we used the fact that $\left(\tfrac { \delta}2\right)^k>2^k\cdot C$, and in (1) we used the fact that for $k=2^r\in\calK$ we have
$$
\log_2\prod_{i\in \calK\cap[k/2]}(k/i)^{i}
=\sum_{i\in \calK\cap[k/2]} i \log (k/i)
=\sum_{j\in \calK\cap[2,k]} (k/j) \log j
=\sum_{i=1}^{r} 2^{r-i} \cdot i < 2^r\cdot 2=2k
$$
since $\sum_{i=1}^{\infty}i2^{-i}=\tfrac 12 f(\tfrac 12)=2$
where $f(z)=\tfrac 1 {(1-z)^2 }=(\tfrac 1{1-z})'=\sum_{i=1}^\infty i z^{i-1}$ for $|z|<1$.

Next, we show~(13).
We have 
$$
\beta_k^2
=K\Delta \cdot 20 \ln \tfrac {16}{ \delta} \cdot k^2(1-\omega^k)
\ge K\Delta \cdot 10 \ln \tfrac {16}{ \delta} \cdot k(k+1)(1-\omega^k)
$$
and so
$$
\exp\left(-\frac{\beta_k^2}{2(k+1)(1-\omega^k)\Delta}\right) 
\le \exp\left(-Kk \cdot 5 \ln \tfrac {16}{ \delta}\right) 
=\left(\tfrac {16}{ \delta} \right)^{-5Kk}
\le \tfrac 12\left(\tfrac {16}{ \delta} \right)^{-4Kk}
\le \tfrac 12 \lambda_{2k}
$$

It remains to show that $\sigma_K\ge \tfrac 1{16} \sigma K$ if $\phi\le \tfrac{1}{40}$.
If $K=1$ then the claim is trivial. Suppose that $K\ge 2$.
Denote $\pi_k=\sigma_k/(k \sigma)$ and $\gamma_k=(1+\omega^k)/2$.
Eq.~(11b) can be written as 
$
2k\sigma\pi_{2k}=2\gamma_k\cdot k\sigma\pi_k-2\phi\sigma k\sqrt{2(1-\gamma_k)}
$, therefore recursions~(11) become
\begin{subequations}
\begin{eqnarray}
\pi_1&=&1 \\
\pi_{2k}&=&\gamma_k\pi_k-\phi \sqrt{2(1-\gamma_k)}\qquad\quad\forall k\in\dot\calK
\end{eqnarray}
\end{subequations}
Denoting  $\Gamma_k=\prod_{i\in\calK\cap[k]}\gamma_i$ and expanding these recursions gives
\begin{equation}\label{eq:GALDSGA}
\pi_{2k}=\Gamma_{k} - \phi\sum_{i\in\calK\cap [k]}\tfrac {\Gamma_{k}}{\Gamma_i}\sqrt{2(1-\gamma_i)}\qquad\quad\forall k\in\dot\calK
\end{equation}

Recall that 
$
e^{-x}>1-x>e^{-x/(1-x)}
$ for $x<1$. From this we get
$$
\omega^k=(1-1/K)^k>e^{-k/(K-1)}>1-\tfrac{k}{K-1}
\qquad\Rightarrow\qquad
\gamma_k>1-\tfrac{k}{2(K-1)}
$$
We claim that $\Gamma_{K/2}\ge \tfrac 18$.~\footnote{Numerical evaluation suggests a better bound: $\Gamma_{K/2}>0.432$.}
Indeed, if $K< 32$ then this can be checked numerically.
Suppose that $K\ge 32$.
We have $\ln(1-x)>-\tfrac{1}{1-x}\cdot x$ for $x\in(0,1)$ and thus
$$
\ln\Gamma_{K/2}
>  \sum_{k\in\dot\calK}\ln \left(1-\tfrac{k}{K-1}\right)
> -\tfrac{1}{1-\tfrac{16}{31}}\sum_{k\in\dot\calK}\tfrac{k}{K-1}
=
 -\tfrac{31}{15}
\qquad\Rightarrow\qquad \Gamma_{K/2}\ge e^{-31/15}>\tfrac 18
$$
Clearly, we have $\Gamma_k\ge \Gamma_{K/2}>\tfrac 18$ for all $k\in\dot\calK$. Therefore,
\begin{align*}
\sum_{i\in \calK\cap[K/2]}\tfrac {\sqrt{2(1-\gamma_i)}}{\Gamma_i} 
&\le 8 \sum_{i\in \calK\cap[K/2]} \sqrt{\tfrac{i}{K-1}}
=\tfrac {8}{\sqrt{K-1}} \cdot (\sqrt {2^0} + \sqrt {2^1} + \sqrt {2^2} + \ldots + \sqrt{K/2}) \\
&=\tfrac {8}{\sqrt{K-1}} \cdot \tfrac {\sqrt{K}-1}{\sqrt 2-1}
<\tfrac {8}{\sqrt 2-1}<20
\end{align*}
Using~\eqref{eq:GALDSGA}, we conclude that if $\phi<\tfrac 1{40}$ then $\pi_K\ge \Gamma_{K/2}(1-20\phi)\ge \tfrac 18\cdot \tfrac 12=\tfrac 1{16}$.
\end{proof}

\subsection{Proof of Theorem~\ref{th:induction}}\label{sec:induction}

To prove the theorem, 
we use induction on $k\in\calK$. For $k=1$ the claim is trivial. 
In the remainder of this section we assume that the claim holds for $k\in\dot\calK$,
and prove it for $2k$ and skew-symmetric matching cover $\calM\subseteq{\tt Matching}$ with ${\tt size}(\calM)\ge \delta_{2k} $.
Clearly, the skew-symmetry of $\calM$ implies that $\calM^k$ is also skew-symmetric.
We introduce the following notation; letter $x$ below always denotes a node in $S$.

\begin{itemize}
%\item Let $\varepsilon=\exp\left(-\tfrac{\beta^2}{2(1-\rho_i)\ell^2}\right)$ and $\xi=\left(\tfrac {\delta_i-\delta_{i-1}}5\right)^{\rho/(1-\rho)}$.
\item Let $\mu_{x}(u,u')$ be the expected out-degree  of $x$  in $M\sim\calM^{k}(u,u')$.
\item Let  $\nu_{x}(u)=\mathbb E_{u'\sim_\rho u}[\mu_{x}(u,u')]$ where $\rho:=\omega^{k-1}$.
\item Let $\calA_x$ be a subset of $\mathbb R^d$ of Gaussian measure $\delta:=\tfrac 12(\delta_{2k}-\delta_k)$
containing vectors $u$ with the largest value of $\nu_{x}(u)$, i.e.\
such that $\gamma_x:=\inf\{\nu_{x}(u)\::u\in \calA_x\}\ge \sup\{\nu_{x}(u)\::u\in \mathbb R^d-\calA_x\}$.

\item Let $\dot\calM\subseteq\calM$ be the matching cover obtained from $\calM$ by removing
edges $(x,y)$ from $\calM(u)$ and edges $(y,x)$ from $\calM(-u)$
for each $u\in \calA_x$.
Clearly, $\dot\calM$ is a skew-symmetric matching cover
with ${\tt size}(\dot\calM)\ge \delta_{2k}-2\delta =\delta_k$.
%By the induction hypothesis, 
%${\tt size}_\rho(\tilde\calM^{(i-1)})\ge \lambda_{i-1}n$.
%\begin{equation}\label{lemma:lambdaSum}
%\sum_{x\in S}\lambda_x={\tt size}(\tilde\calM^{(i-1)})\ge \lambda_{i-1}n
%\end{equation}
\item Let $\dot\mu_x(u,u')$ be the expected out-degree of $x$  in $M\sim\dot\calM^{k}(u,u')$.
\item Let  $\dot\nu_x(u)=\mathbb E_{u'\sim_\rho u}[\dot\mu_{x}(u,u')]$, and let $\lambda_x=||\dot\nu||_1$.
Note that $\tfrac 1n\sum_{x\in V}\lambda_x = {\tt size}_\rho(\dot\calM^k)\ge \lambda_k $
where the last inequality is by the induction hypothesis.
%\item Let $\tau_x$ be the expected number of paths through $x$ in $M\sim\calM^k\circ_{\rho,\omega,\rho}^\sigma\calM^k$

%under the following process:
%$(v',v,u,u')\sim\calN_{\rho,\omega,\rho}$
%\item Let $\lambda_x=\mathbb E_{\hat u\sim\calN^d,u_1\sim_\rho\hat u,u_2\sim_\rho\hat u}\mu_{x\ast}(u_1,\hat u,u_2)
%=\mathbb E_{\hat u\sim\calN^d,u_1\sim_\omega\hat u,u_2\sim_\omega\hat u}\mu_{\ast x}(u_1,\hat u,u_2)$,
%where the last equality is by skew-symmetry.
%Clearly, we have $||\nu_{x\ast}||_1=\mathbb E_{u_1} [\mathbb E_{\hat u\sim_\omega u_1,u_2\sim_\omega \hat u}[\mu_{x\ast}(u_1,\hat u,u_2)]]=\lambda_x$
%and similarly $||\nu_{\ast x}||_1=\lambda_x$.

\end{itemize}

\begin{lemma}
$\gamma_x\ge \lambda_x$.
\end{lemma}
\begin{proof}
By the definition of $\dot\calM$, we have $\dot\mu_x(u,u')=0$ for any $u\in\calA_x$ and $u'\in\mathbb R^d$.
This implies that  $\dot\nu_x(u)=0$ for any $u\in\calA_x$.
We have $\dot\calM\subseteq\calM$ and thus $\dot\calM^k\subseteq\calM^k$.
This implies that $\dot\nu_x(u)\le \nu_x(u)\le\gamma_x$ for any $u\in\mathbb R^d-\calA_x$.
We can now conclude that $\lambda_x=||\dot \nu_x||_1\le \gamma_x\cdot 1$.

\end{proof}

Let $\calH$ be an extended matching cover where $\calH(u'_1,u'_2)$ is defined
as follows: 
\begin{itemize}
\item[$(\ast)$]
sample $(u'_1,u_1,u_2,u'_2)\sim\calN_{\rho,\omega,\rho}|(u'_1,u'_2)$,
sample $M_1\sim\calM^k(u'_1,u_1)$, sample $M_2\sim\calM^k(u_2,u'_2)$, compute $M=M_1\circ M_2$,
output $M^{\tt good}={\tt Truncate}_{\sigma_{2k}}(M;u'_1)\cap {\tt Truncate}_{\sigma_{2k}}(M;u'_2)$.
\end{itemize}
It can be seen that  $\calH\subseteq\calM^{2k}$ under a natural coupling.
%Next, we will analyze ${\tt size}_{\omega^{2k-1}}(\calH)={\tt size}_{\rho\omega\rho}(\calH)$.
%For that we need to 
Consider the following process: sample $(u'_1,u'_2)\sim\calN_{\rho\omega\rho}$
and then run procedure $(\ast)$.
By definitions, we have ${\tt size}_{\rho\omega\rho}(\calH)=\tfrac 1n\mathbb E[|M^{\tt good}|]$.
Let  $M^{\tt good}_x$ be the set of paths in $M$ that  go through node $x$,
and denote $\tau_x=\mathbb E[|M^{\tt good}_x|]$. From these definitions we get that
 ${\tt size}_{\rho\omega\rho}(\calH)=\tfrac 1n\sum_{x\in V}\tau_x$.

\begin{lemma}
$\tau_x\ge\delta^{2/(1-\omega)}\lambda_x^2-2\varepsilon$ where $\varepsilon:=\exp\left(-\frac{\beta_k^2}{2(k+1)(1-\omega\rho)\Delta}\right)$.
\end{lemma}
\begin{proof}
Define random variable $p_1$ and $p_2$ as follows: 
\begin{itemize}
\item if $x$ has an incoming path in $M_1$ then let $p_1$ be this path, otherwise let $p_1=\perp$; 
\item if $x$ has an outgoing path in $M_1$ then let $p_2$ be this path, otherwise let $p_2=\perp$. 
\end{itemize}
Let $\calE^{\tt good}(p_1,p_2)=[p_1 \ne \perp \wedge\;p_2 \ne \perp]$.
Let $\calE^{\tt bad}_1(p_1,u'_2)$ be the event that $p_1=(y_1,\ldots,x)\ne\perp$, $p_1$ is nonviolating and $\langle x-y_1,u'_2\rangle<\sigma:=\omega\rho\sigma_k-\beta_k$.
Similarly, let 
$\calE^{\tt bad}_2(p_2,u'_1)$ be the event that $p_2=(x,\ldots,y_2)\ne\perp$, $p_2$ is nonviolating and $\langle y_2-x,u'_1\rangle<\sigma$.
Note, if  $[\calE^{\tt good}(p_1,p_2)\wedge \neg \calE^{\tt bad}_1(p_1,u'_2)]$ holds
then $p_1\circ p_2=(y_1,\ldots,x,\ldots, y_2)$ is either violating or satisfies
$\langle y_2-y_1,u'_2\rangle \ge \sigma_k + \sigma =\sigma_{2k}$
(and similarly for $\calE^{\tt bad}_2(p_2,u'_1)$).
Therefore, $\tau_x\ge {\tt Pr}[\calE^{\tt good}(p_1,p_2)]-{\tt Pr}[\calE^{\tt bad}_1(p_1,u'_2)] - {\tt Pr}[\calE^{\tt bad}_2(p_2,u'_1)]$.

Clearly, vectors $(u'_1,u_1,u_2,u'_2)$
are distributed according to $\calN_{\rho,\omega,\rho}$.
Equivalently, they are obtained by the following process:
sample $(u_1,u_2)\sim\calN_\omega$, sample $u'_1\sim_\rho u_1$, sample $u'_2\sim_\rho u_2$.
Define $\calB_x=\{u\::\:-u\in\calA_x\}$.
By Theorem~\ref{th:isoperimetric}, we will have $(u_1,u_2)\in\calA_x\times\calB_x$ with
probability at least $\delta^{2/(1-\omega)}$.
Conditioned on the latter event, we have ${\tt Pr}[p_1\ne \perp]\ge\gamma_x$ and
${\tt Pr}[p_2\ne \perp]\ge\gamma_x$ (independently), where the claim for $p_1$ follows from the skew-symmetry of $\calM^k$.
This implies that ${\tt Pr}[\calE^{\tt good}(p_1,p_2)]\ge \delta^{2/(1-\omega)}\gamma_x^2\ge \delta^{2/(1-\omega)}\lambda_x^2$.

We claim that ${\tt Pr}[\calE^{\tt bad}_1(p_1,u'_2)]\le \varepsilon$.
Indeed, it suffices to prove that for fixed $u'_1,u_1,p_1$ we have 
${\tt Pr}[\calE^{\tt bad}_1(p_1,u'_2)]\le \varepsilon$ under $u_2\sim_\omega u_1$, $u'_2\sim_\rho u_2$
(or equivalently under $u'_2\sim_{\omega\rho} u_1$).
Assume that $p_1=(y_1,\ldots,x)\ne\perp$ is nonviolating (otherwise the desired probability is zero and the claim holds).
Since $\calM^k$ is $\sigma_k$-stretched, we have $\langle x-y_1,u_1\rangle \ge \sigma_k$.
We also have $r:=||x-y_1||\le \sqrt{(k+1)\Delta}$ by Lemma~\ref{lemma:regular}(a).
The quantity $\langle x-y_1,u'_2\rangle$ is normal with mean $\omega\rho \langle x-y_1,u_1\rangle\ge \omega\rho\sigma_k$
and variance $(1-(\omega\rho)^2)r^2$ % under $u\sim_\omega \hat u$
%$\sim r\cdot\calN(\omega,(1-\omega^2)^{1/2})$ for $u\sim_\omega\hat u$
(since e.g.\ we can assume that $x-y_1=(r,0,\ldots,0)$ by rotational symmetry, and then use the definitions of correlated Gaussians
for 1-dimensional case).
Therefore, 
\begin{align*}
{\tt Pr}_{ u'_2\sim_{\omega\rho} u_1}\left[\langle x-y_1,u'_2\rangle<\sigma\right] 
&\le {\tt Pr}\left[\calN(\omega\rho\sigma_k,(1-(\omega\rho)^2)^{1/2}r)<\omega\rho\sigma_k-\beta_k\right] \\
&= {\tt Pr}\left[\calN(0,1)<-\frac{\beta_k}{(1-(\omega\rho)^2)^{1/2}r}\right]  
<\exp\left(-\frac{\beta_k^2}{2(1-(\omega\rho)^2)r^2}\right)\le\varepsilon
\end{align*}

In a similar way we prove that ${\tt Pr}[\calE^{\tt bad}_2(p_2,u'_1)]\le \varepsilon$. The lemma follows.

\end{proof}

We showed that 
$$
{\tt size}_{\rho\omega\rho}(\calH)\ge \left(\tfrac 1n\sum_{x\in V} \delta^{2/(1-\omega)}\lambda_x^2 \right) - 2\varepsilon
$$
Let us minimize the bound on the RHS under constraint $\tfrac 1n\sum_{x\in V}\lambda_x\ge \lambda_k $.
Clearly, the minimum is obtained when $\lambda_x=\lambda_k$ for all $x\in V$,
in which case the bound becomes
$$
{\tt size}_{\rho\omega\rho}(\calH)\ge  \delta^{2/(1-\omega)}\lambda_k^2 - 2\varepsilon\ge\lambda_{2k}
$$
where we used \eqref{eq:MKASASF} and \eqref{eq:HDHGDHDGAHKA}.

%%
%% If your work has an appendix, this is the place to put it.
\appendix

\section{Matrix multiplicative weights (MW) algorithm}\label{sec:app:MW}
It this section we review the method of Arora and Kale~\cite{AK}
for the checking the feasibility of system (\ref{eq:SDP}$'$)
consisting of constraints \eqref{eq:SDP:a'} and \eqref{eq:SDP:b}-\eqref{eq:SDP:e}.
The algorithm is given below.

\begin{algorithm}[H]
  \DontPrintSemicolon
\SetNoFillComment
\For{$t=1,2,\ldots,T$}
{
	compute 
$$
W^{(t)}=\exp\left(\eta\sum_{r=1}^{t-1}N^{(r)}\right)\;,\quad
X^{(t)} = 
n \cdot \frac{W^{(t)}} {{\tt Tr}(W^{(t)})}
$$
\\
either output {\tt Fail} or find ``feedback matrix'' $N^{(t)}$ of the form
 $$
 N^{(t)}={\tt diag}(y)+\sum_p f_pT_p + \sum_S z_S K_S - F
 $$
%with $||M^{(t)}||\le 1$ such that $M_t \bullet X^{(t)}>0$ and $M_t\bullet X\le 0$ for all feasible $X$
%(or output ``fail''). Matrix $M^{(t)}$ should have the form
 where $f_p\ge 0$, $z_S\ge 0$, $\sum_i y_i+\xi n^2\sum_Sz_S\ge \alpha$, $F$ is a symmetric matrix with $F\preceq C$, and $N^{(t)}\bullet X^{(t)}\le 0$.
}

      \caption{MW algorithm. 
      }\label{alg:MW}
\end{algorithm}

%\begin{theorem}[\cite{AK}]
%For any sequence of loss matrices $M^{(1)},\ldots,M^{(T)}$, the matrices $X^{(1)},\ldots,X^{(T)}$ generated by the algorithm
%satisfy
%\begin{equation}
%\frac{1}{n}\sum_{t=1}^T M^{(t)}\bullet X^{(t)}\le\lambda_n\left(\sum_{t=1}^T M^{(t)} \right)+\frac{\eta}n\sum_{t=1}^T(M^{(t)})^2\bullet X^{(t)}+\frac{\ln n}\eta
%\end{equation}
%\end{theorem}
It can be checked that $N^{(t)}\bullet X\ge 0$ for all feasible $X$ satisfying (\ref{eq:SDP}$'$) (see~\cite{AK}). Thus, matrix $N^{(t)}$ can be viewed as a cutting plane
that certifies that  $X^{(t)}$ is infeasible (or lies on the boundary of the feasible region).

The procedure at line 3 is called {\tt Oracle}, and the maximum possible spectral norm $||N^{(t)}||$ of the feedback matrix is called the {\em width} of the oracle.
This width will be denoted as $\rho$.~\footnote{
Note that~\cite{AK} formulated the algorithm in terms of the ``loss matrix'' $M^{(t)}=-\frac 1\rho N^{(t)}$
that satisfies $||M^{(t)}||\le 1$.
Namely, it used the update $W^{(t)}\!=\!\exp\left(-\bar\eta\sum_{r=1}^{t-1}M^{(r)}\right)$
with $\bar\eta\!=\!\rho\eta$, and set $\bar\eta\!=\!\frac{\varepsilon}{2\rho n}$
in their Theorems 4.4 and 4.6.
}

\begin{theorem}[{\cite[Theorems 4.4 and 4.6]{AK}}]
Set $\eta=\frac{\epsilon}{2\rho^2 n}$ and $T=\lceil \frac{4\rho^2 n^2 \ln n}{\epsilon^2}\rceil$.
If Algorithm~\ref{alg:MW} does not fail during the first $T$ iterations then
the optimum value of SDP~\eqref{eq:SDP} is at least $\alpha-\epsilon$.
\end{theorem}

The oracle used in~\cite{AK} has the following property: if it fails then it returns a cut which is $\frac{c}{512}$-balanced
and has value at most $\kappa \alpha$, where value $\kappa$ depends on the implementation. One of the implementations
achieves $\kappa=O(\sqrt{\log{n}})$ and has width $\rho=\tilde O(\frac \alpha n)$. (The runtime of this oracle will be discussed later).
Setting $\epsilon=\alpha/2$ yields an $O(\kappa)=O(\sqrt{\log{n}})$ approximation algorithm
that makes $\tilde O(1)$ calls to the oracle.

\subsection{Gram decomposition and matrix exponentiation}\label{sec:Gram}
Consider matrix $X=X^{(t)}$ computed at the $t$-th step of Algorithm~\ref{alg:MW}.
It can be seen that $X$ is positive semidefinite, so we can consider its Gram decomposition: $X=V^TV$.
Let $v_1,\ldots,v_n$ be the columns of $V$; clearly, they uniquely define $X$. 
Note that computing $v_1,\ldots,v_n$ requires matrix exponentiation, which is a tricky operation because of accuracy issues.
Furthermore, even storing these vectors requires $\Theta(n^2)$ space and thus $\Omega(n^2)$ time, which is too slow for our purposes.
To address these issues, Arora and Kale compute approximations $\tilde v_1,\ldots,\tilde v_n$ to these vectors
using the following result.
\begin{theorem}[{\cite[Lemma 7.2]{AK}}]\label{th:GramApproximation}
For any constant $c>0$ there exists an algorithm that does the following:
given values $\gamma\in(0,\tfrac 12)$, $\lambda>0$, $\tau=O(n^{3/2})$ and  matrix $A\in\mathbb R^{n\times n}$ of spectral norm $||A||\le \lambda$, it computes 
matrix $\tilde V\in\mathbb R^{d\times n}$ with column vectors
$\tilde v_1,\ldots,\tilde v_n$ 
of dimension $d=O(\tfrac{\log n}{\gamma^2})$ such that matrix $\tilde X=\tilde V^T\tilde V$ has trace $n$,
and with probability at least $1-n^{-c}$,
one has
\begin{subequations}\label{eq:GramApproximation}
\begin{align}
|\;||\tilde v_i||^2-||v_i||^2\;| &\;\;\le\;\; \gamma (||\tilde v_i||^2+\tau) && \forall i \\
|\;||\tilde v_i-\tilde v_j||^2-||v_i-v_j||^2\;| &\;\;\le\;\; \gamma (||\tilde v_i-\tilde v_j||^2+\tau) && \forall i,j
\end{align}
\end{subequations}
where $v_1,\ldots,v_n$ are the columns of a Gram decomposition of $X=n\cdot \frac{\exp(A)}{{\tt Tr}(\exp(A))}$.
The complexity of this algorithm equals the complexity of computing $kd$ matrix-vector products of the form $A\cdot u$, $u\in\mathbb R^n$,
where $k=O(\max\{\lambda^2,\log \frac{n^{5/2}}{\tau}\})$.
\end{theorem}
Note that matrices $A$ used in Algorithm~\ref{alg:MW} have norm at most $\eta\rho T$.
Therefore, we can set $\lambda=\Theta(\frac{\rho n \log n}{\alpha})$
when applying Theorem~\ref{th:GramApproximation} to Algorithm~\ref{alg:MW}.
Parameters $\gamma$ and $\tau$ will be specified later.

From now on we make the following assumption.
\begin{assumption}\label{assumption:one}
We have (unobserved) matrix $V\in\mathbb R^{n\times n}$ and (observed) matrix $\tilde V\in\mathbb R^{d\times n}$ with 
$X=V^TV$, $\tilde X=\tilde V^T\tilde V$ and
${\tt Tr}(X)={\tt Tr}(\tilde X)=n$
satisfying conditions~\eqref{eq:GramApproximation} where $v_1,\ldots,v_n$ are the columns of $V$ and $\tilde v_1,\ldots,\tilde v_n$ are the columns of $\tilde V$.
\end{assumption}

\subsection{Oracle implementation}
Let us denote $S=\{i\in V\::\:||\tilde v_i||^2\le 2\}$. 
We have $\sum_{i\in V}||\tilde v_i||^2={\tt Tr}(\tilde V^T\tilde V)=n$ and thus $|S|\ge n/2$.
First, one can eliminate an easy case.
\begin{proposition}\label{prop:oracle-easy-case}
Suppose that $K_S\bullet \tilde X< \tfrac {\xi n^2}4$.
Then setting $y_i=-\tfrac{\alpha}n$ for all $i\in V$, $z_{S}=\tfrac{2\alpha}{\xi n^2}$, $z_{S'}=0$ for all $S'\ne S$, and $F=0$
gives a valid output of the oracle with width $\rho=O(\tfrac \alpha n)$ assuming that 
 parameters $\tau,\gamma$ in Theorem~\ref{th:GramApproximation} satisfy 
 $\gamma\le \tfrac 12$ and $\tau\le \tfrac \xi 2$.
\end{proposition}
\begin{proof}
Denote $z_{ij}=||v_i-v_j||^2$ and $\tilde z_{ij}=||\tilde v_i-\tilde v_j||^2$.
We know that $\tilde Z:=\sum_{ij} \tilde z_{ij}=K_{S}\bullet \tilde X<\tfrac {\xi n^2}4$ where the sum is over $i,j\in S$.
Also, $|z_{ij}-\tilde z_{ij}|\le \gamma(\tilde z_{ij}+\tau)$ for all $i,j$.
This implies that 
$$
K_{S}\bullet X=\sum_{ij} z_{ij}
< \tilde Z + \gamma \tilde Z + \tfrac{n^2}{2}\gamma\tau
\le(1+\gamma) \tfrac {\xi n^2}4 + \tfrac{n^2}{2} \gamma\tau 
\le \tfrac{\xi n^2}2
$$
Note that $N^{(t)}=-\tfrac \alpha n I + \tfrac {2\alpha}{\xi n^2}K_{S}$.
We have $\sum_iy_i+\xi n^2\sum_Sz_S=n\cdot(-\tfrac \alpha n)+\xi n^2\cdot \tfrac {2\alpha}{\xi n^2}=\alpha$
and $N^{(t)}\bullet X=-\tfrac \alpha n I\bullet X + \tfrac {2\alpha}{\xi n^2}(K_{S}\bullet X)
\le -\tfrac \alpha n \cdot n + \tfrac {2\alpha}{\xi n^2}\cdot \tfrac {\xi n^2}2=0$, as desired.
Also, $||N^{(t)}||\le||-\tfrac \alpha n I||+||\tfrac {2\alpha}{\xi n^2}K_{S}||=O(\tfrac \alpha n)$.
\end{proof}

Now suppose that $\sum_{i,j\in S}||\tilde v_i-\tilde v_j||^2=K_S\bullet \tilde X\ge \tfrac{\xi n^2}4$.
We will set $N^{(t)}_{ij}=0$ for all $(i,j)\notin S\times S$.
When describing how to set the remaining entries of $N^{(t)}$,
it will be convenient to treat $N^{(t)}$, $X$, $\tilde X$ etc.\ as matrices of size
$|S|\times |S|$ rather than $n\times n$, and $y$ as a vector of size $|S|$.
(Note that $X$ and $\tilde X$ are the submatrices of the original matrices, and ${\tt Tr}(X)\le n$).
With some abuse of terminology we will redefine $V=S$ and $n=|S|$,
and refer to the original variables as $V_{\tt orig}$ and $n_{\tt orig}\in[n,2n]$.
Thus, from now on we make the following assumption.
\begin{assumption}\label{assumption:two}
$||\tilde v_i||^2\le 2$ for all $i\in V$ and $\sum_{i,j\in V:i<j}||\tilde v_i-\tilde v_j||^2\ge \tfrac {\xi n_{\tt orig}^2}4\ge \tfrac {\xi n^2}4$.
\end{assumption}

Let us set $y_i=\frac \alpha n$ for all $i\in V$ and $z_S=0$ for all $S$, then $\sum_i y_i+\xi n_{\tt orig}^2\sum_S z_S=\alpha$.
Note that $N^{(t)}=\tfrac \alpha n I + \sum_p f_p T_p - F$ and $\tfrac \alpha n I\bullet X\le\alpha$,
so condition $(\sum_p f_p T_p - F)\bullet X\le -\alpha$
would imply that $N^{(t)}\bullet X\le 0$.
Our goal thus becomes as follows.

\vspace{5pt}
\noindent\hspace{0pt}
\begin{minipage}{\dimexpr\columnwidth-10pt\relax}
\fbox{\parbox{\textwidth}{
 \em Find variables $f_p\ge 0$ and symmetric matrix $F\preceq C$ such that $(\sum_p f_p T_p - F)\bullet X\le -\alpha$.
}}
\end{minipage}
\vspace{5pt}

We arrive at the specification of {\tt Oracle} given in Section~\ref{sec:MW}.

\section{Proof of Theorem~\ref{th:isoperimetric}}\label{sec:isoperimetric}
Here we repeat the argument in~\cite{Mossel06}.
For a function $f:\mathbb R^d\rightarrow\mathbb R_{\ge 0}$
denote
$||f||_p=(\mathbb E_u[f^p(u)])^{1/p}$. Let $T_\omega f$ is the function $\mathbb R^d\rightarrow\mathbb R_{\ge 0}$
defined via $(T_\omega f)(u)=\mathbb E_{\hat u\sim_\omega u} f(\hat u)$. ($T_\omega$ is called the ``Ornstein-Uhlenbeck operator'').
\begin{theorem}[{\cite{Borell}}]\label{th:Borell}
Let $f:\mathbb R^d\rightarrow\mathbb R_{\ge 0}$ and $-\infty<q\le p\le 1$, $0\le\omega\le\sqrt{\frac{1-p}{1-q}}$. Then
$$
||T_\omega f||_q \ge ||f||_p
$$
\end{theorem}
Let $f,g$ be the indicator functions of sets $\calA,\calB$ respectively.
Set $p=1-\omega$ and $q=1-1/\omega$, so that $1/p+1/q=1$ and $\sqrt{\tfrac{1-p}{1-q}}=\omega$.
Note that $||f||_p=||g||_p=\delta^{1/p}$. We can write
\begin{align*}
{\tt Pr}_{(u,u')\sim\calN_\omega}[(u,u')\in \calA\times \calB]
&=\mathbb E_{(u,u')\sim\calN_\omega}[f(u)g(u')] 
=\mathbb E_{u}[f T_\omega g]  \\
&\ge ||f||_p ||T_\omega g||_q && \mbox{\em \small by reverse H\"older inequality} \\
&\ge ||f||_p ||g||_p && \mbox{\em \small by Theorem~\ref{th:Borell} } \\
& = \delta^{2/(1-\omega)}
\end{align*}

\bibliographystyle{alpha}
\bibliography{maxflow}

\end{document}